\documentclass[12pt]{article}

\usepackage{mathrsfs}
\usepackage{appendix}
\usepackage{natbib}
\usepackage{extarrows}
\usepackage{graphicx,setspace,lscape,longtable,epstopdf,xr}
\usepackage{mathrsfs,amsmath,amsthm,amssymb,color}
\usepackage{epsfig,graphicx,pdfpages}
\usepackage{rotating}
\usepackage{float}
\usepackage{bm}
\usepackage{enumerate}
\usepackage{ragged2e}
\usepackage{todonotes}
\usepackage[hidelinks]{hyperref}
\usepackage{multirow}
\hypersetup{colorlinks=true, linkcolor=blue, citecolor=blue, urlcolor=blue}
\usepackage{booktabs}
\usepackage{algorithm}
\usepackage{algpseudocode}
\usepackage{amsmath}
\usepackage{graphicx}
\usepackage{subfigure}

\usepackage{hyperref}
	
\bibpunct{(}{)}{;}{a}{,}{,}

\newtheorem{theorem}{Theorem}
\newtheorem{lemma}{Lemma}
\newtheorem{proposition}{Proposition}

\renewcommand\tablename{Table}

\def\one{{\bf 1}}
\def\bA{{\mathbf A}}
\def\ba{{\mathbf a}}
\def\bomega{{\boldsymbol{\omega}}}

\def\ind{\textup{ind}}
\def\sh{\textup{sh}}
\def\tr{\mbox{tr}}

\def\beq{\begin{equation}}
\def\eeq{\end{equation}}
\def\beqr{\begin{eqnarray}}
\def\eeqr{\end{eqnarray}}
\def\beqrs{\begin{eqnarray*}}
\def\eeqrs{\end{eqnarray*}}
\def\bet{\begin{theorem}}
\def\eet{\end{theorem}}
\def\bel{\begin{lemma}}
\def\eel{\end{lemma}}
\def\bep{\begin{proposition}}
\def\eep{\end{proposition}}
\def\bg{\begin{figure}[tbph]\begin{center}}
\def\eg{\end{center}\end{figure}}

\def\bc{\begin{center}}
\def\ec{\end{center}}

\def\BIC{\mbox{BIC}}
\def\TPR{\mbox{TPR}}
\def\FPR{\mbox{FPR}}
\def\SD{\mbox{SD}}
\def\BIC{\mbox{BIC}}

\def\wt{\widetilde}

\def\wh{\widehat}
\def\wb{\overline}

\def\bb{\mathbf{b}}
\def\bB{\mathbf{B}}
\def\mA{\mathbb A}

\def\C{\mathbb C}
\def\bD{\mathbb D}
\def\mN{\mathcal{N}}

\def\mR{\mathbb{R}}

\def\mS{\mathbb S}

\def\H{\mathbb H}

\def\bM{\mathbb M}

\def\bS{\mathbf S}
\def\mS{\mathcal S}
\def\mE{\mathcal E}

\def\bM{\mathbf{M}}
\def\mT{\mathcal T}
\def\bH{\mathbb H}

\def\V{\mathbb{V}}

\def\mZ{\mathbb{Z}}
\def\bx{\mathbf{x}}
\def\bZ{\mathbf{Z}}
\def\bV{\mathbf{V}}
\def\bu{\mathbf{u}}
\def\bbf{\mathbf{f}}
\def\bF{\mathbf{F}}

\def\bv{\mathbf{v}}
\def\var{\mbox{var}}
\def\supp{\mbox{supp}}

\def\cov{\mbox{cov}}
\def\lasso{\textup{lasso}}
\def\oracle{\textup{oracle}}

\def\ols{\text{OLS}}
\def\initial{\textup{initial}}
\def\argmin{\mbox{argmin}}

\def\diag{\mbox{diag}}

\def\CS{\mbox{CS}}
\def\Bias{\mbox{Bias}}

\def\rmse{\mbox{RMSE}}
\def\vec{\mbox{vec}}

\newcommand{\bbeta}{\boldsymbol{\beta}}

\def\bSigma{\boldsymbol{\Sigma}}
\def\bOmega{\boldsymbol{\Omega}}

\def\bdelta{\boldsymbol{\delta}}
\def\bDelta{\boldsymbol{\Delta}}

\newcommand{\RNum}[1]{\uppercase\expandafter{\romannumeral #1\relax}}
\def\bB{\mathbf{B}}
\def\bC{\mathbf{C}}
\def\bE{\mathbf{E}}
\def\bX{\mathbf{X}}
\def\bx{\mathbf{x}}
\def\bY{\mathbf{Y}}
\def\bZ{\mathbf{Z}}

\def\by{\mathbf{y}}
\def\bc{\mathbf{c}}
\def\bV{\mathbf{V}}
\def\bG{\mathbf{G}}
\def\bH{\mathbf{H}}
\def\bD{\mathbf{D}}

\def\bW{\mathbf{W}}
\def\bw{\mathbf{w}}
\def\bu{\mathbf{u}}

\def\bg{\mbox{\boldmath $g$}}
\def\bI{\mathbf{I}}

\def\zero{\mathbf{0}}
\def\defeq{\stackrel{\mathrm{def}}{=}}  
\def\sign{\mbox{sign}}

\def\boxit#1{\vbox{\hrule\hbox{\vrule\kern6pt\vbox{\kern6pt#1\kern6pt}\kern6pt\vrule}\hrule}}

\numberwithin{equation}{section}

\begin{document}
\begin{center}

{\bf\Large Penalized Sparse Covariance Regression with High Dimensional Covariates}\\
\bigskip

\iftrue{  
	
	Yuan Gao$^1$, Zhiyuan Zhang$^{2}$, Zhanrui Cai$^{4}$,
	Xuening Zhu$^{2,3*}$, \\
	Tao Zou$^{5}$ and Hansheng Wang$^{1}$
	
	{\it \small
		$^1$Guanghua School of Management, Peking University, Beijing China;\\
		$^2$School of Data Science, Fudan University, Shanghai, China;\\
		$^3$MOE Laboratory for National Development and Intelligent Governance, Fudan University, Shanghai, China;\\
		$^4$Faculty of Business and Economics, The University of Hong Kong, Hong Kong, China;\\
		$^5$Research School of Finance, Actuarial Studies and Statistics, Australian National University, Canberra, Australia
	}
	
}\fi
\end{center}

\begin{footnotetext}[1]{
	Xuening Zhu ({\it xueningzhu@fudan.edu.cn}) is the corresponding author.}
	\end{footnotetext}

\begin{singlespace}
\begin{abstract}

  Covariance regression offers an effective way to model the large covariance matrix with the auxiliary similarity matrices.
  In this work, we propose a sparse covariance regression (SCR) approach to handle the potentially high-dimensional predictors (i.e., similarity matrices).
  Specifically, we use the penalization method to identify the informative predictors and estimate their associated coefficients simultaneously.
  We first investigate the Lasso estimator and subsequently consider the folded concave penalized estimation methods (e.g., SCAD and MCP).
  However, the theoretical analysis of the existing penalization methods is primarily based on \textit{i.i.d.} data, which is not directly applicable to our scenario.
  To address this difficulty, we establish the non-asymptotic error bounds by exploiting the spectral properties of the covariance matrix and similarity matrices.
  Then, we derive the estimation error bound for the Lasso estimator and establish the desirable oracle property of the folded concave penalized estimator.
  Extensive simulation studies are conducted to corroborate our theoretical results. We also illustrate the usefulness of the proposed method by applying it to a Chinese stock market dataset.


\end{abstract}
\end{singlespace}

\newpage

\section{Introduction}

Estimating the covariance matrix is an essential task for many statistical learning problems. For instance, for financial risk management, the covariance matrix estimated from the stock returns can be used to construct investment portfolios \citep{goldfarb2003robust,fan2012volatility,fan2012vast}.
In network data analysis, estimating the covariance matrix of the associated responses is helpful to understand the network structure  \citep{lan2018covariance, liu2020semiparametric}.
In addition, for many popular multivariate statistical methods like linear discriminant analysis (LDA),
the estimation of the covariance matrix is often a prerequisite operation \citep{johnson1992applied, pan2016ultrahigh}. Therefore, obtaining a reliable estimate of the covariance matrix is of great importance.

The main challenge of the covariance matrix estimation is that the number of unknown parameters can be huge, especially for large-scale covariance matrix \citep{bickel2008regularized, fan2016overview}.
To deal with this issue, two common approaches exist in the literature.
The first approach assumes a sparse or a low-rank structure for the covariance matrix \citep{bickel2008covariance,bickel2008regularized,lam2009sparsistency,cai2011adaptive, fan2011high,fan2013large,fan2018large}.
Consequently, specific regularization algorithms can be applied to recover the covariance matrix's intrinsic sparsity or low-rank structure.
However, this approach typically requires many repeated observations of the response vectors to obtain a reliable estimation result.
As an alternative approach, \cite{zou2017covariance} proposes a covariance regression framework, directly expressing the covariance matrix as a linear combination of known similarity matrices.
The similarity matrices can be constructed from auxiliary covariates or network structures among the subjects.
Take the stock returns as an example.
To estimate the covariance matrix for the stock returns,
we can collect a number of firms' fundamentals
as the auxiliary information.
In addition, we can use the industrial information and common shareholder relationship among the stocks
to construct networks.
One can easily construct many similarity matrices from the above auxiliary and network information.
This enables us to obtain a reliable estimation for the large-scale covariance matrix, especially when the number of periods is limited.

Despite the usefulness of the covariance regression model, its performance can be unstable when a large number of predictors (i.e., the similarity matrices) are available.
That is because estimating many regression coefficients simultaneously in the covariance regression model is challenging.
To deal with the potential high dimensionality of regression coefficients, a popular solution is to
impose the sparsity assumption on the coefficients \citep{fan2001variable,fan2004nonconcave,wang2009shrinkage}, which enables us to select the predictors with significant contributions.
Meanwhile, it allows us to obtain a more reliable estimate for the covariance matrix.

To achieve this goal, we consider using penalized estimation methods in the covariance regression model.
For the conventional regression models, the $L_1$-penalized (i.e., Lasso) regression \citep{tibshirani1996regression} is widely used due to its computational attractiveness and good performance in practice.
However, it has been shown that the Lasso estimator requires relatively strong conditions to achieve the variable selection consistency \citep{zou2006adaptive,zhao2006model}.
The folded concave penalized methods, such as SCAD \citep{fan2001variable} and MCP \citep{zhang2010nearly}, are proposed to achieve the desirable oracle property under milder conditions.
Namely, they could estimate the nonzero regression coefficients as if we knew the true sparsity pattern in advance.
The folded concave penalized regression model has been extensively studied in recent years \citep{fan2011nonconcave,zhang2012general,wang2013calibrating,fan2014strong,fan2017high}.
Various research studies \citep{wang2007tuning,zou2008one,fan2011sparse,zhu2020nonconcave} also illustrate its theoretical and practical advantages.

Although these penalized methods for conventional regression models have been well studied, to our best knowledge, they have not yet been applied to the covariance regression model discussed in this study.
The traditional regression model typically assumes that the data are independent and identically generated from the same underlying model\citep{fan2001variable,wang2013calibrating,fan2014strong}, or follow certain dependence structures, such as time series \citep{chan2014group}. However, the previous situations are distinctly different from the covariance regression model considered in the current paper. Although we can treat the covariance regression model as a particular type of matrix regression, it is important to note that the matrix entries are not independently distributed but have special dependence structures. The new structure presents significant challenges in deriving the estimation error bound, especially in high-dimensional settings.

This paper studies the properties of the penalized estimation methods for the sparse covariance regression (SCR) model.
To demonstrate the advantages of the SCR model, we first consider the most challenging situation where only a single observation of the response is available.
We investigate the Lasso estimator and derive the corresponding non-asymptotic error bound. The results demonstrate that the Lasso estimator is consistent, but unfortunately its oracle property is not guaranteed.
To address this limitation, we explore the folded concave penalized estimation method.
Specifically, we use the Lasso estimator as the initial value for the local linear approximation (LLA) algorithm to compute its solution.
Theoretically, we establish the strong oracle property for the resulting estimator, indicating that the LLA algorithm can converge exactly to the oracle estimator with an overwhelming probability.
Moreover, we demonstrate the asymptotic normality for the oracle estimator in a more general case.
Lastly, we extend the SCR model to the scenario with repeated observations of the response.
In this case, faster convergence rate can be obtained and heterogeneity can be well accommodated.
We also demonstrate that the SCR model can be naturally combined with the classical factor models.
This leads to a new class of factor composite models with better modeling flexibility.
We then apply those methods to analyze the returns of the stocks traded in the Chinese A-share market with encouraging feedback.

The rest of the article is organized as follows.
In Section 2, we introduce the penalized regression methods for the sparse covariance regression (SCR) model.
Section 3 investigates the theoretical properties of the proposed estimators.
Section 4 explores some extensions for the scenario involving repeated observations.
Numerical studies are given in Section 5.
Finally, we provide all technical proof details and additional numerical experiments in the Appendix.

\section{Sparse Covariance Regression}

\subsection{Model and Notations}

Let $\by = (Y_1,\cdots, Y_p)^\top \in \mR^p$ be a continuous $p$-dimensional vector
with mean $\zero$ and covariance $\bSigma = E(\by\by^\top)\in \mR^{p \times p}$.
In addition, for the $j$th subject, we collect a set of associated covariates
as $\bx_j = (X_{j1},\cdots , X_{jK})^\top\in \mR^{K}$.
For example, $Y_j$ can be the stock return of the $j$th firm,
and $\bx_j$ is the associated the financial fundamentals (e.g., market value, cash flow).

To model the covariance matrix $\bSigma$, we follow \cite{zou2017covariance} to consider a set of similarity matrices.
First, the similarity matrix can be constructed based on the covariate information $\bx_j\ (1\le j\le p)$.
Suppose the $k$th type of covariate is a continuous variable, then the similarity between the subject $j_1$ and $j_2$ can be defined as $w_{k,j_1j_2} = \exp\{-d(X_{j_1k}, X_{j_2k})\}$, where $d(X_{j_1k}, X_{j_2k})$ denotes certain type of distance function between
$X_{j_1k}$ and $ X_{j_2k}$. For a discrete covariate, the similarity between subject $j_1$ and $j_2$ can be defined if they share the same value.
For instance, in a stock network, we define
\begin{align*}
w_{k,j_1j_2} = \left \{
\begin{array}{cl}
  1 & \mbox{if the stocks $j_1$ and $j_2$ are in the same industry}  \\
  0 & \mbox{otherwise}
\end{array}
 \right.
\end{align*}
In social network analysis, the similarity matrix can also be defined by the friend relationships among the network users.
Then we express the covariance matrix by a linear combination of the similarity matrices, i.e.,
\begin{equation}
\bSigma(\bbeta) = \beta_0 \bI_p + \sum_{k = 1}^K\beta_k \bW_k,\label{eq:Sigma_0}
\end{equation}
where $\bW_k = (w_{k,j_1j_2}) \in \mR^{p \times p}$
is the similarity matrix constructed based on $k$th covariate $\bX_k = (X_{1k},\cdots, X_{pk})^\top \in \mR^p$.
Here $\beta_k$s ($0\le k\le K$) are corresponding covariance regression coefficients.
Note that similarity matrices typically have the same diagonal elements. For example, when using continuous covariates $\bX_k$s to construct similarity matrices as described above, all their diagonal elements are equal to $\exp(0)=1$. In this case, the model can be rewritten as $\bSigma(\bbeta) =  \sum_{k=0}^K \beta_k \bI_p + \sum_{k = 1}^K\beta_k (\bW_k-\bI_p)$. Then the diagonal elements of $\bW_k-\bI_p$ become zeros for each $1\le k\le K$. Therefore, for the similarity matrices $\bW_k\ (1\le k\le K)$ with the the same diagonal elements, we set them to be zeros as suggested by \cite{zou2017covariance}.
However, when $\bW_k$s have different diagonal elements, we can leave the diagonal elements of $\bW_k$s as they are. The numerical studies in Section \ref{subsec:empirical} and Appendix A.7 present some concrete examples.
Let $\bbeta^{(0)} = (\beta_0^{(0)},\cdots ,\beta_K^{(0)})^\top$ be the true regression vector of $\bbeta$ in \eqref{eq:Sigma_0} and we consider a sparse structure of $\bbeta^{(0)}$.
Specifically, let $\mS = \supp(\bbeta^{(0)})$ collects the indexes of nonzero coefficients.
Consequently, we have $\beta_k^{(0)}\ne 0$ for $k\in \mS$ and $\beta_k^{(0)} = 0$ for $k\not\in \mS$.
Given \eqref{eq:Sigma_0}, the sparse covariance regression (SCR) model can be expressed as
\begin{equation*}
\by\by^\top  = \beta_0 \bI_p+ \sum_{k = 1}^K \beta_k\bW_k + \mE,
\end{equation*}
where $\mE$ is a symmetric random matrix that satisfies $E(\mE) = \zero_{p\times p}$.
Without loss of generality, we let $\bW_0 = \bI_p$ in the following, and denote $\bSigma_0 = \bSigma(\bbeta^{(0)}) \defeq \sum_{k=0}^K \beta_k^{(0)} \bW_k$ as the true covariance matrix.

{\sc Notation.}
Throughout this paper, we denote the cardinality of a set $\mS$
by $|\mS|$.
In addition, let $\mS^c$ complement the set $\mS$.
For a vector $\bv = (v_1,\cdots, v_p)^\top\in\mR^p$, let
$\|\bv\|_q = (\sum_{j = 1}^p v_j^q)^{1/q}$ for $q>0$.
For convenience, we omit the subindex when $q = 2$.
Denote $\supp(\bv)$ as the support of the vector.
Particularly, we use $\|\bv\|_{\infty}$ to denote $\max_{j} |v_j|$, and $\| \bv\|_{\min}$ to denote $\min_{j} |v_j|$.
In addition, denote $\bv_{\mS}$ as a sub-vector of $\bv$ as
$\bv_{\mS} = (v_j: j\in\mS)^\top\in\mR^{|\mS|}$.
For symmetric matrix $\bA = (a_{ij}) \in \mR^{p\times p}$, we use $\lambda_{\max}(\bA)$ and $\lambda_{\min}(\bA)$ to denote the maximum and minimum eigenvalues of $\bA$, respectively.
For an arbitrary matrix $\bM = (m_{ij}) \in \mR^{p_1\times p_2}$, denote $\|\bM\|=\|\bM\|_2 = \lambda^{1/2}_{\max}(\bM^\top\bM)$, $\|\bM\|_1 = \max_{1\le j\le p_2}(\sum_{i = 1}^{p_1}|m_{ij}|)$, $\|\bM\|_\infty = \max_{1\le i\le p_1}(\sum_{j = 1}^{p_2}|m_{ij}|)$,
and $\| \bM\|_F = \left( \sum_{i, j} m_{ij}^2 \right)^{1/2}$.
For arbitrary two sequences $\{a_N\}$ and $\{b_N\}$, denote $a_N\gg b_N$ to mean that $a_N/b_N\rightarrow\infty$.

\subsection{Penalized Estimation}

To estimate the coefficients of the covariance regression model, \cite{zou2017covariance} proposed to use a least squares objective function,
\begin{equation}\label{eq:Qbeta}
    Q(\bbeta) = \frac{1}{2p} \Big\|\by\by^\top - \bSigma(\bbeta) \Big\|_F^2.
\end{equation}
Let $ \wh{\bbeta}_{\ols}= \arg\min Q(\bbeta)$ be the ordinary least squares (OLS) solution to \eqref{eq:Qbeta}. Then one can derive its analytical form as
\begin{equation}\label{eq:ols_est}
\wh{\bbeta}_{\ols}= \arg\min Q(\bbeta) =\bSigma_W^{-1}\bSigma_{WY},
\end{equation}
where $\bSigma_W = \{\tr(\bW_k \bW_l): 0\le k,l\le K\}\in \mR^{(K+1)\times (K+1)}$
and $\bSigma_{WY} = \{\by^\top \bW_k \by: 0\le k\le K\}^\top \in \mR^{K+1}$.
The OLS estimation is feasible when $K$ is of low dimension.
However, if the number of candidate similarity matrices is large,
one cannot obtain a reliable estimator of $\bbeta$ using the OLS method.

Considering the high dimensionality of the problem and the sparsity of the regression coefficients, we first consider the Lasso penalized estimator for the sparse covariance regression (SCR) model as follows:
\begin{equation}\label{eq:lasso}
    \wh \bbeta^{\lasso} = \argmin_{\bbeta} Q(\bbeta) + \lambda_0  \|\bbeta\|_1,
\end{equation}
where $Q(\cdot)$ is defined in \eqref{eq:Qbeta}, and $\lambda_0\ge 0$ is a tuning parameter. With $\lambda_0 = 0$, the estimator reduces to the OLS estimator as \eqref{eq:ols_est}.
In practice, if we have the preliminary information that some predictors (i.e., $\bW_k$s) are important, we can directly keep the corresponding coefficients unpenalized. For example, the intercept $\beta_0$ corresponding to $\bW_0 = \bI_p$ is usually left out of the penalty term.
To compute the Lasso estimator in \eqref{eq:lasso}, efficient algorithms like LARS \citep{efron2004least} and coordinate descent \citep{friedman2007pathwise} can be implemented. However, the Lasso estimator is not guaranteed to possess oracle property in general \citep{zou2006adaptive}.

To address this issue, we adopt the folded concave penalized SCR method. Specifically, we need to minimize the following penalized loss function as
\begin{equation}\label{eq:Q_lambda}
Q_\lambda(\bbeta) = Q(\bbeta)+ \sum_{k = 0}^K p_\lambda(|\beta_k|),
\end{equation}
where $p_\lambda(\cdot)$ is the folded concave penalty function and $\lambda\ge 0$ is a tuning parameter. 
Following \cite{fan2014strong}, throughout the article, we assume that the folded concave penalty function $p_\lambda(|t|)$ defined on $t \in (-\infty, \infty)$ satisfies:
\begin{enumerate}[(i)]
    \item $p_\lambda(t)$ is increasing and concave in $t\in[0, \infty)$ with $p_\lambda(0)=0$;
    \item $p_\lambda(t)$ is differentiable in $t\in (0, \infty)$ with derivative $p'_\lambda(0) \defeq p'_\lambda(0+)\ge a_1 \lambda$;
    \item $p'_\lambda(t)\ge a_1 \lambda$ for $t\in (0, a_2 \lambda]$;
    \item $p'_\lambda(t)=0 $ for $t\in[\gamma\lambda, \infty)$ with the prespecified constant $\gamma>a_2$.
\end{enumerate}
Here, $a_1$ and $a_2$ are two fixed positive constants.
The above definition includes and extends the popularly used SCAD penalty \citep{Fan:Li:2001} and MCP penalty \citep{zhang2010nearly}.
The SCAD penalty function takes the form as
\begin{align*}
    p_{\lambda, \gamma}(t) = \begin{cases}
      \lambda t & \text{if  $0\le t \le \lambda$,}\\
      \frac{2\gamma\lambda t - (t^2+\lambda^2)}{2(\gamma - 1)} & \text{if $\lambda < t \le \gamma \lambda$,}\\
      \frac{\lambda^2(\gamma^2 - 1)}{2(\gamma - 1)} & \text{if $t> \gamma\lambda$,}
    \end{cases}
\end{align*}
for some $\gamma >2$.
The MCP penalty function takes the form as
\begin{align*}
    p_{\lambda, \gamma}(t) = \begin{cases}
    \lambda t - \frac{t^2}{2\gamma} & \text{if $0\le t \le \gamma \lambda$,}\\
    \frac{1}{2} \gamma \lambda^2 & \text{if $t > \gamma \lambda$,}
    \end{cases}
\end{align*}
for some $\gamma>1$.
It is easy to verify that $a_1 = a_2 = 1$ for the SCAD penalty, and $a_1 = 1 - \gamma^{-1}$, $a_2 = 1$ for the MCP penalty, according to the previous definition.
We visualize the two penalty functions in Figure \ref{fig:scad_mcp}.

\begin{figure}[htbp]
	\centering
	\subfigure{
		\begin{minipage}[t]{0.5\linewidth}
			\centering
			\includegraphics[scale=0.4]{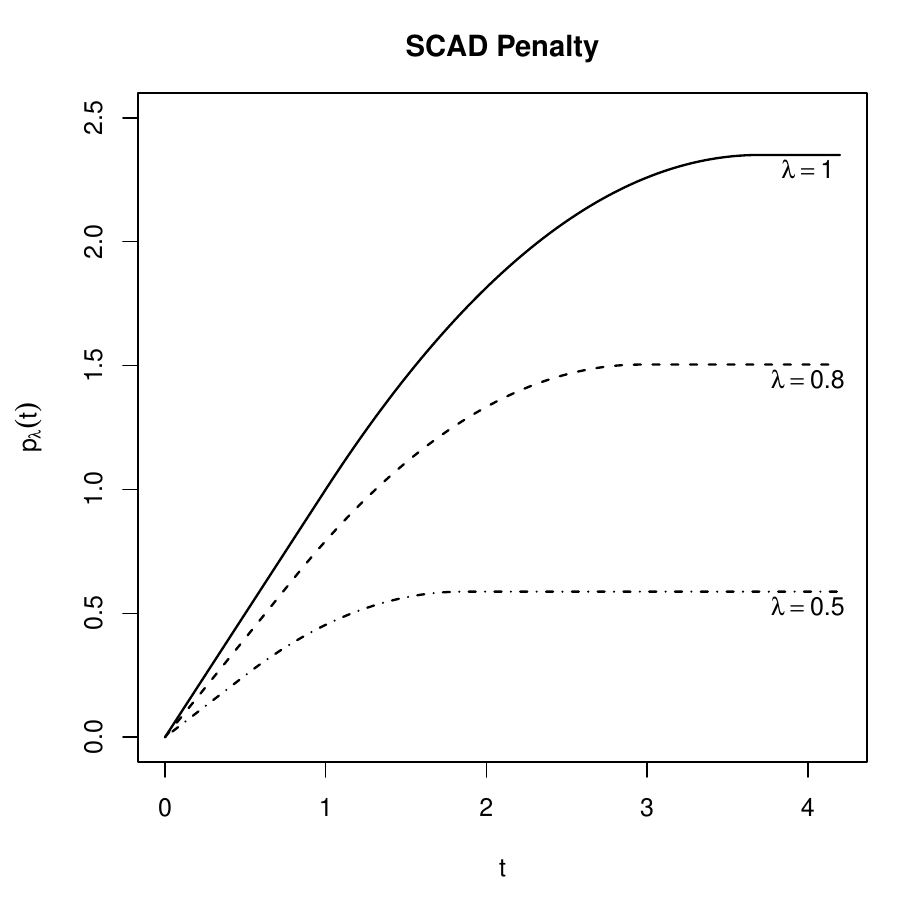}
		\end{minipage}%
	}%
	\subfigure{
		\begin{minipage}[t]{0.5\linewidth}
			\centering
			\includegraphics[scale=0.4]{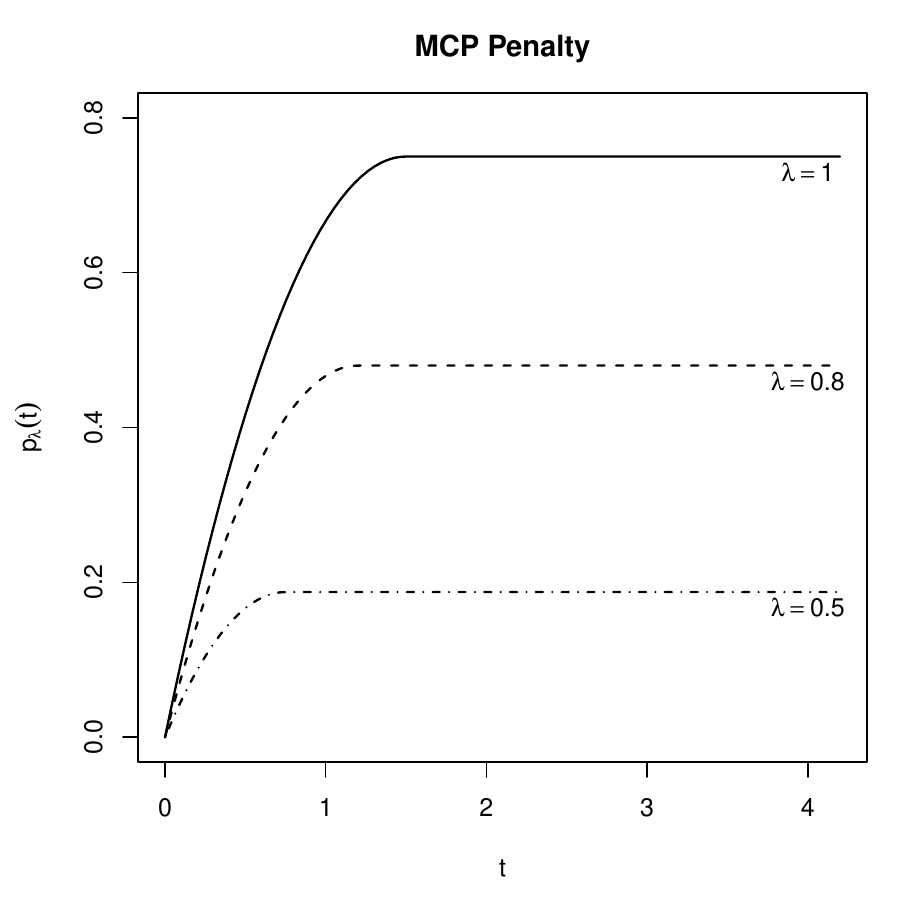}
		\end{minipage}
	}
	\centering
	\caption{The SCAD ($\gamma = 3.7$) and MCP ($\gamma = 1.5$) penalty functions with different values of $\lambda$.}\label{fig:scad_mcp}
\end{figure}

The local linear approximation (LLA) algorithm \citep{zou2008one} is adopted to minimize the objective function defined in \eqref{eq:Q_lambda}.
The algorithm details are summarized in Algorithm \ref{alg:LLA}.
To implement the LLA algorithm, an initial estimator $\wh\bbeta^{\initial}$ needs to be specified.
It can be observed that if the LLA algorithm is initialized by zero, then the one-step
estimator should be the solution to $\argmin_{\bbeta} \big\{Q(\bbeta) + p_\lambda'(0) \|\bbeta\|_1 \big\}$. This is actually a Lasso estimation problem equivalent to \eqref{eq:lasso}.
Consequently, we use the Lasso estimator $\wh \bbeta^{\lasso}$ to initialize the LLA algorithm.
In the next section, we investigate the theoretical properties of the Lasso estimator and the resulting estimator of the LLA algorithm.

\begin{algorithm}
\caption{The local linear approximation (LLA) algorithm}\label{alg:LLA}
\begin{enumerate}

    \item Initialize $\wh \bbeta^{(0)} = \wh \bbeta^{\initial}$, and compute the adaptive weights:
    \begin{equation*}
        \wh\bw^{(0)} = \Big(\wh w_0^{(0)}, \dots, \wh w_K^{(0)}\Big)^\top = \Big(p'_\lambda(| \wh\beta_0^{(0)}|), \dots, p'_\lambda(| \wh\beta_K^{(0)}|)\Big)^\top.
    \end{equation*}
    \item For $m=1,2,\dots$, repeat the LLA iteration till converge
    \begin{enumerate}[(2.a)]
        \item Obtain $\wh \bbeta^{(m)}$ by solving the following optimization problem:
        \begin{equation*}
            \wh \bbeta^{(m)} = \argmin_{\bbeta} Q(\bbeta) +  \sum_{k = 0}^K \wh w_j^{(m-1)}|\beta_k|;
        \end{equation*}
        \item Update the adaptive weight vector $\wh\bw^{(m)}$ with $\wh w_k^{(m)} = p'_\lambda(| \wh\beta_k^{(m)}|)$ for $0\le k \le K$.
    \end{enumerate}
\end{enumerate}
\end{algorithm}

\section{Theoretical Properties}

Recall that $\mS = \supp(\bbeta^{(0)})$ collects the indexes of nonzero coefficients of the true coefficient $\bbeta^{(0)}$. 
Without loss of generality, we assume $\mS = \{0,1,\dots, s\}$ with $  |\mS|=s+1>0$. Obviously, the complement of $\mS$ should be $\mS^c = \{s+1, \dots, K\}$.
If we know the true support set $\mS$ in advance, then we can define the oracle estimator for SCR model as
\begin{equation}\label{eq:oracle}
     \wh \bbeta^{\oracle} = (\wh \bbeta^{\oracle\top}_{\mS}, \zero^\top)^\top = \argmin_{\bbeta: \bbeta_{\mS^c=\zero}} Q(\bbeta),
\end{equation}
where $Q(\bbeta)$ is the unpenalized loss defined in \eqref{eq:Qbeta}.
Similar to \eqref{eq:ols_est}, we can compute that $\bbeta^{\oracle}_{\mS} = \bSigma_{W,\mS}^{-1} \bSigma_{WY, \mS}$,
provided $\bSigma_{W,\mS}$ is invertible. Here, $ \bSigma_{W, \mS} = \{\tr(\bW_k\bW_l):k,l\in \mS \} \in \mR^{(s+1)\times (s+1)}$ and $\bSigma_{WY,\mS} = (\by^\top\bW_k\by: k\in \mS )^\top \in \mR^{s+1}$.

To facilitate the theoretical investigation, we specify some technical conditions as follows.
\begin{enumerate} [(C1)]
    \item ({\sc Minimal Signal Strength}) Assume $\| \bbeta_{\mS }^{(0)}\|_{\min} > (\gamma+1)\lambda $. \label{cond:minimal_signal}

    \item({\sc Minimal Eigenvalue}) Assume that $\inf_p \lambda_{\min}(p^{-1} \bSigma_{W, \mS}) \ge \tau_{\min}$ holds for some positive constant $\tau_{\min}$, where $\bSigma_{W, \mS} = \{\tr(\bW_k\bW_l):k,l\in \mS \} \in \mR^{(s+1)\times (s+1)}$.    \label{cond:minimal_eigen}

    \item ({\sc Sub-Gaussian Distribution}) Assume $\by = \bSigma_0^{1/2}\bZ$ with $\bZ=(Z_1,\dots, Z_p)^\top \in \mR^{p}$, where $Z_j$'s are independent and identically distributed mean zero sub-Gaussian random variables, that is, $E(e^{t Z_j}) \le e^{c^2 t^2 / 2},\ \forall t$ for some constant $c>0$. Further assume that, for each $1\le j\le p$, $\var(Z_j)=1$ and $E(Z_j^4) = \mu_4$.
    In addition, we assume that there exists a positive constant $\sigma_{\min}$ such that $\inf_p \lambda_{\min}(\bSigma_0)>\sigma_{\min}$.\label{cond:distribution}

    \item ({\sc Bounded $\ell_1$-Norm}) For all symmetric matrices in $\{\bW_k\in\mR^{p\times p}: 0\le k \le K \}$, there exists $w>0$ such that $\sup_{p,k} \|\bW_k\|_1 \le w < \infty $.
    Further assume that $\sup_{p} \|\bSigma_0^{1/2}\|_1 \le \sigma_{\max}^{1/2}$ for some finite positive constant $\sigma_{\max}$.
    \label{cond:bounded_norm}

    \item ({\sc Restricted Eigenvalue}) Define the set $\C_3(\mS) \defeq \{\bdelta \in \mR^{K+1}: \|\bdelta_{\mS^c}\|_1 \le 3 \|\bdelta_{\mS}\|_1\}$.
    Assume $\{\bW_k\}_{0\le k \le K}$ satisfies the restricted eigenvalue (RE) condition, that is,
    \begin{equation*}
        \frac{1}{p} \left\|\sum_{k=0}^K \delta_k \bW_k \right\|_F^2 \ge  \kappa \| \bdelta\|^2, \quad \text{for all $\bdelta \in \C_3(\mS)$}
    \end{equation*}
    for some constant $\kappa>0$.\label{cond:RE}

    \item ({\sc Convergence}) Assume that (i) $\bG_{d,p}\defeq
    p^{-1} \big\{ \tr(\bSigma_0^d \bW_k\bSigma_0^d \bW_l): k,l\in \mS \big\}$ converges to a positive definite matrix $\bG_d \in \mR^{(s+1)\times (s+1)}$ for $d=0,1$ in the Frobenius norm, that is,
    $
        \| \bG_{d,p} - \bG_d\|_F \to 0 \text{ as $p \to \infty$},
    $
    where $\bSigma_0^0\defeq\bI_p$.
    Furthermore, assume $\lambda_{\min}(\bG_d)\ge \tau_0$ for some finite positive constant $\tau_0$; (ii) $\bH_p \defeq p^{-1} \big\{\tr[(\bSigma_0^{1/2}\bW_k\bSigma_0^{1/2})\circ (\bSigma_0^{1/2}\bW_l \bSigma_0^{1/2})]:k,l\in \mS \big\}$ converges to a matrix $\bH \in\mR^{(s+1)\times (s+1)}$ in Frobenius norm, where $\circ$ denotes the Hadamard product. \label{cond:convergence}

\end{enumerate}
We comment on these conditions in the following.
Condition (C\ref{cond:minimal_signal}) imposes a constraint on the minimum signal strength of the nonzero coefficients, which is necessary for establishing the oracle property. Similar conditions have been commonly used in previous literature on sparse regression; see for example \cite{fan2004nonconcave}, \cite{wang2013calibrating}, and \cite{fan2014strong}.
Condition (C\ref{cond:minimal_eigen}) ensures that the oracle estimator in \eqref{eq:oracle} is uniquely defined.
By this condition, the informative similarity matrices $\bW_k$s $(0\le k\le s)$ should not be severely correlated with each other.
Condition (C\ref{cond:minimal_eigen}) has been rigorously and theoretically verified by an important example in Appendix A.6. 
Condition (C\ref{cond:distribution}) assumes a sub-Gaussian distribution condition on the response variable.
This condition is necessarily needed for deriving some non-asymptotic probability bounds by the Hanson-Wright type inequality. 
The additional minimal eigenvalue condition in Condition (C\ref{cond:distribution}) ensures the positive definiteness of $\bSigma_0$.
Condition (C\ref{cond:bounded_norm}) imposes bounded $\ell_1$-norm condition on matrices $\bW_k$s and $\bSigma_0$, which implies the bounded operator norm conditions as assumed in \cite{zou2017covariance}.
This condition is helpful for deriving the non-asymptotic probability bounds and establishing the asymptotic normality.
We can also allow the upper bound $w$ to slowly diverge to infinity as $p\to\infty$ at an appropriate rate. Then more sophisticated theoretical treatments are needed.
Condition (C\ref{cond:RE}) is a restricted eigenvalue (RE) type condition, which is used to derive the $\ell_2$-error bound for the Lasso estimator.
Condition (C\ref{cond:RE}) has been theoretically verified by an important example in Appendix A.6. 
Lastly, Condition (C\ref{cond:convergence}) is a law of large numbers type assumption, which is used to form the asymptotic covariance matrix of the oracle estimator.
Similar conditions are imposed in \cite{zou2017covariance} and \cite{zou2022inference}.
Condition (C\ref{cond:convergence}) has also been theoretically verified for a special case in Appendix A.6. 

We first give the error bound for the Lasso estimator in the following theorem.

\begin{theorem} 
\label{thm:lasso_est}
Assume Conditions (C\ref{cond:distribution})--(C\ref{cond:RE}). Then $ \| \wh\bbeta^{\lasso} - \bbeta^{(0)}\| \le (3/\kappa)\sqrt{s+1} \lambda_0$ holds with probability at least $1-\delta_0'$, where
\begin{align*}
    \delta_0' = 2(K+1)\exp\left\{ -\min\left( \frac{C_1 p \lambda_0^2}{w^2\sigma_{\max}^2},\frac{C_2 p \lambda_0}{w\sigma_{\max}} \right)\right\},
\end{align*}
and $C_1, C_2$ are two positive constants.
\end{theorem}

The proof of Theorem \ref{thm:lasso_est} is given in the Appendix. 
From Theorem \ref{thm:lasso_est} we can see that, if $K$ is fixed and we take $\lambda_0 = C_0 p^{-1/2} $ for some positive constant $C_0$, then we should have $\| \wh\bbeta^{\lasso} - \bbeta^{(0)}\| = O_p(p^{-1/2})$. In other words, the Lasso estimator is $\sqrt{p}$-consistent in the finite parameter setting, which aligns with the results in \cite{zou2017covariance}.
By this result, we can find that the dimension $p$ here plays a role like ``sample size'' as in the conventional regression models.
The larger $p$ we have, the more information we collect, and then the more accurate estimator can be obtained.
We then use the Lasso estimator as the initial estimator for the LLA algorithm to compute the folded concave penalized estimator.
The properties of the LLA algorithm and the resulting estimator are given in the following theorem.

\begin{theorem}
\label{thm:convergence}
Assume Conditions (C\ref{cond:minimal_signal}) and (C\ref{cond:minimal_eigen}). Then the LLA algorithm initialized by $\wh \bbeta^{\initial} $ converges to $\wh \bbeta^{\oracle}$ after two iterations with probability at least $1 - \delta_0 - \delta_1 - \delta_2$, where $\delta_0 = P\Big(\| \wh \bbeta^{\initial} - \bbeta^{(0)} \|_{\infty} > a_0 \lambda\Big)$, $\delta_1 = P\Big(\|\nabla_{\mS^c} Q(\wh \bbeta^{\oracle}_{\mS})\|_{\infty}\ge a_1 \lambda\Big)$, $\delta_2 = P\Big(\|\wh \bbeta^{\oracle}_{\mS} \|_{\min} \ge \gamma\lambda\Big)$, and $a_0 = \min\{1, a_2\}$. Moreover, $a_1, a_2, \gamma$ are constants specified in (i)--(iv).
Suppose we use Lasso estimator $\wh\bbeta^{\lasso}$ as the initial estimator and pick $\lambda\ge (3\sqrt{s+1}\lambda_0 ) / (a_0 \kappa)$. Further assume Conditions (C\ref{cond:distribution})--(C\ref{cond:RE}).
Then, it holds that
\begin{align*}
    & \delta_0 \le 2(K+1)\exp\left\{ -\min\left( \frac{C_1 p \lambda_0^2}{w^2\sigma_{\max}^2},\frac{C_2 p \lambda_0}{w\sigma_{\max}} \right)\right\}, \\
    & \delta_1 \le 2 (K-s) \exp\left\{-\min\left(\frac{C_3 a_1^2 p \lambda^2 }{w^2\sigma_{\max}^2 }, \frac{C_4 a_1 p \lambda}{ w\sigma_{\max}}\right) \right\} \\ &~~~~~+2(K-s)(s+1)\exp\left[-\min\left\{\frac{C_5 a_1^2 \tau_{\min}^2 p \lambda^2}{w^6\sigma_{\max}^2(s+1)^2}, \frac{C_6 a_1 \tau_{\min}p \lambda}{w^3 \sigma_{\max}(s+1)} \right\} \right],\\
     &  \delta_2 \le 2 (s+1)\exp \left[-\min\left\{\frac{C_7 \tau_{\min}^2 p (\| \bbeta_{\mS}^{(0)}\|_{\min} - \gamma \lambda)^2 }{w^2\sigma_{\max}^2 (s+1)}, \frac{C_8 \tau_{\min}  p(\| \bbeta_{\mS}^{(0)}\|_{\min} - \gamma \lambda) }{w\sigma_{\max}(s+1)^{1/2}} \right\} \right],
\end{align*}
where $C_1,\dots,C_8$ are some positive constants.
In particular, if $p\lambda_0^2 / \{s\log(K)\} \to \infty$, then we have $\delta_0 + \delta_1 + \delta_2 \to 0$ as $p \to \infty$.
\end{theorem}

The proof of Theorem \ref{thm:convergence} is given in the Appendix. 
From Theorem \ref{thm:convergence}, we can see that, if we use Lasso estimator as the initial estimator, then the LLA algorithm can converge exactly to the oracle estimator with overwhelming probability under appropriate conditions. This property is referred to as the strong oracle property in \cite{fan2014strong}.
In addition, if we take $\lambda = (3\sqrt{s+1}\lambda_0 ) / (a_0 \kappa)$, then $p\lambda_0^2 / \{s\log(K)\} \to \infty$ is equivalent to $\lambda \gg s \sqrt{\log(K) / p}$.
Consequently, to fulfill $\| \bbeta_{\mS }^{(0)}\|_{\min} > (\gamma+1)\lambda $ in Condition (C\ref{cond:minimal_signal}), we require that $K = o\big(\exp(p \|\bbeta^{(0)}\|_{\min}^2 / s^2) \big)$.
We remark that this is not a very stringent requirement.
For example, if $s$ is fixed and the minimal signal $\| \bbeta_{\mS }^{(0)}\|_{\min} > c$ for some constant $c>0$, then the number of similarity matrices (i.e., $K$) is allowed to diverge in a rate extremely close to $O(\exp(p))$.
Further note that the strong oracle property implies the resulting estimator of the LLA algorithm should have the same asymptotic distribution as the oracle estimator \citep{Fan:Li:2001}.
In this regard, we establish the asymptotic normality of the oracle estimator in the following theorem.

\bet 
\label{thm:asymptotic_normality}
Assume Conditions (C\ref{cond:minimal_eigen})--(C\ref{cond:bounded_norm}) and (C\ref{cond:convergence}).
Let $\bA \in \mR^{L \times (s+1)}$ be an arbitrary matrix with $\sup_{s} \|\bA\|<\infty$,  where $L>0$ is a fixed integer.
Suppose (i) $(s+1)^{-1} \bA \{2\bG_1 + (\mu_4 -3) \bH\} \bA^\top \to \bC$ if $s\to \infty$ or (ii) $\bC \defeq (s+1)^{-1} \bA \{2\bG_1 + (\mu_4 -3) \bH\} \bA^\top $ if $s$ is fixed, where $\bC \in\mR^{L\times L}$ is a positive definite matrix.
Then we have,
\begin{align*}
    \sqrt{p/(s+1)}\bA\bG_0 \Big(\wh\bbeta^{\oracle}_{\mS} - \bbeta_{\mS}^{(0)}\Big) \to_d \mN\big(\zero,\bC\big),\text {as $p\to\infty$.}
\end{align*}
\eet

The proof of Theorem \ref{thm:asymptotic_normality} is given in the Appendix. 
This theorem generalizes the result in \cite{zou2017covariance} by allowing diverging feature dimension $s$ and relaxing the normal distribution assumption.
In fact, if $s$ is fixed and $\by$ follows $\mN(\zero, \bSigma_0)$, we can take $\bA = \bI_{s+1}$. Then we should have $\sqrt{p} \big(\wh\bbeta^{\oracle}_{\mS} - \bbeta_{\mS}^{(0)}\big)  \to_d \mN\big(\zero, 2\bG_0^{-1} \bG_1 \bG_0^{-1}\big) $. This result echoes Theorem 2 in \cite{zou2017covariance}.
On the other hand, if $s$ is diverging as $p\to \infty$, one can take $\bA$ to be any appropriate matrix for finite dimension projection. Then we should have $\sqrt{p/(s+1)}\bA\bG_0 \big(\wh\bbeta^{\oracle}_{\mS} - \bbeta_{\mS}^{(0)}\big)$ is asymptotically normal. By Theorem \ref{thm:convergence}, we know that the resulting estimator of the LLA algorithm should enjoy the same asymptotic properties as the oracle estimator under the regularity conditions.

\section{Some Extensions for Repeated Observations}

\subsection{SCR Model for Repeated Observations}
\label{subsec:scr_n}
In the previous sections, we focus on the case where $n=1$ and $p$ tends to infinity.
In practice, we often encounter the situations, where repeated observations of the response vector can be obtained.
Then, how to use all these observations to improve the estimation accuracy of the SCR model becomes an important problem.
We first remark that model \eqref{eq:Sigma_0} implies a homogeneous variance structure of $\bSigma$, since the similarity matrices $\bW_k\ (0\le m\le K)$ typically have the same diagonal elements. In fact, we can allow for a heterogeneous variance structure by replacing the identity matrix $\bI_p$ with a general diagonal matrix $\bD = \diag \{\sigma_1^2, \dots, \sigma_p^2\} $, if the diagonal matrix $\bD$ is known as a prior knowledge.
However, when $\bD$ is unknown, repeated observations are inevitably needed for consistently estimating the heterogeneous variance structure.
Specifically, with repeated observations $\{Y_{ji}: 1\le i \le n \}$ for each $1\le j\le p$, we are able to estimate $\var(Y_{ji}) = \sigma_j^2$ by $\wh \sigma_j^2 = n^{-1}\sum_{i=1}^n (Y_{ji} - \wb Y_j)^2 $, where $\wb Y_j = n^{-1}\sum_{i=1}^n Y_{ji}$. Next, we can standardize $Y_{ji}$ as $\wt Y_{ji} = (Y_{ji}-\wb Y_j) / \wh \sigma_j$ so that the equal variance assumption implied by \eqref{eq:Sigma_0} holds approximately.
Subsequently, we should always assume that $Y_{ji}$s have been standardized appropriately so that model \eqref{eq:Sigma_0} holds.
We need to remark that the homogeneous variance structure of $\bSigma$ is an assumption for technical convenience. With the help of this assumption, we might show that the $\wh\bbeta_n^{\lasso} $ is $\sqrt{np}$-consistent with a fixed $K$ as in the following Theorem \ref{thm:lasso_est_n}.  However, if the estimation errors of those variances estimator $\wh \sigma_j^2 $ are taken into consideration, the conclusions become questionable and need to be further investigated.

We next consider how to extend our results to $n\to \infty$.
Specifically, let $\by_i\ (1\le i \le n)$ be the $n$ independent and identically distributed response vectors.
Then we can modify the original least squares objective function in \eqref{eq:Qbeta} to be $Q_n(\bbeta) = (2np)^{-1} \sum_{i=1}^n \big\|\by_i\by_i^\top - \bSigma(\bbeta) \big\|_F^2$.
Similarly, we use the LLA algorithm to find the solution to the following folded concave penalized loss function $Q_{n,\lambda}(\bbeta) = Q_n(\bbeta)+ \sum_{k = 0}^K p_\lambda(|\beta_k|)$.
Note that the only modification needed for Algorithm \ref{alg:LLA} is to replace $Q(\bbeta)$ with $Q_n(\bbeta)$.
We still use the Lasso penalized estimator $\wh \bbeta_n^{\lasso} = \argmin_{\bbeta} Q_n(\bbeta) + \lambda_0  \|\bbeta\|_1$ as the initial estimator for the LLA algorithm. The error bound for the Lasso estimator is given in the following theorem.

\begin{theorem} 
\label{thm:lasso_est_n}
Assume Conditions (C\ref{cond:distribution})--(C\ref{cond:RE}). Then $\| \wh\bbeta_n^{\lasso} - \bbeta^{(0)}\| \le (3/\kappa)\sqrt{s+1} \lambda_0$
holds with probability at least $1-\delta_0'$, where
\begin{align*}
    \delta_0' = 2(K+1)\exp\left\{ -\min\left( \frac{C_1 n p \lambda_0^2}{w^2\sigma_{\max}^2},\frac{C_2 n p \lambda_0}{w\sigma_{\max}} \right)\right\},
\end{align*}
and $C_1, C_2$ are two positive constants.
\end{theorem}
The proof of Theorem \ref{thm:lasso_est_n} is given in the Appendix. 
Compared with Theorem \ref{thm:lasso_est}, we find that $\wh\bbeta_n^{\lasso} $ is $\sqrt{np}$-consistent for $\bbeta^{(0)}$, if $K$ is fixed and $\lambda_0 = C_0 (np)^{-1/2}$ for some positive constant $C_0$. This indicates that a faster convergence rate can be achieved with repeated observations.
Note that the oracle estimator is defined as $\wh \bbeta_n^{\oracle} = (\wh \bbeta_{n,\mS}^{\oracle\top}, \zero^\top)^\top = \argmin_{\bbeta: \bbeta_{\mS^c=\zero}} Q_n(\bbeta)$.
We next summarize the properties of the LLA algorithm in the following theorem, whose proof is given in the Appendix. 
Compared with Theorem \ref{thm:convergence}, we can find that the main difference is the factor $p$ in the probability upper bounds is replaced by $np$.
This indicates that the LLA algorithm can still converge to the oracle estimator with high probability. Then we can expect that the resulting estimator should be $\sqrt{np}$-consistent when $K$ is fixed.

\bet
\label{thm:convergence_n}
Assume Conditions (C\ref{cond:minimal_signal})--(C\ref{cond:RE}).
Suppose we use Lasso estimator $\wh\bbeta_n^{\lasso}$ as the initial estimator and pick $\lambda\ge (3\sqrt{s+1}\lambda_0 ) / (a_0 \kappa)$.
Then the LLA algorithm converges to $\wh \bbeta_n^{\oracle}$ after two iterations with probability at least $1 - \delta_0 - \delta_1 - \delta_2$ with
\begin{align*}
    & \delta_0 \le 2(K+1)\exp\left\{ -\min\left( \frac{C_1 n p \lambda_0^2}{w^2\sigma_{\max}^2},\frac{C_2 n p \lambda_0}{w\sigma_{\max}} \right)\right\}, \\
    & \delta_1 \le 2 (K-s) \exp\left\{-\min\left(\frac{C_3 a_1^2 n p \lambda^2 }{w^2\sigma_{\max}^2 }, \frac{C_4 a_1 n p \lambda}{ w\sigma_{\max}}\right) \right\} \\ &~~~~~+2(K-s)(s+1)\exp\left[-\min\left\{\frac{C_5 a_1^2 \tau_{\min}^2 n p \lambda^2}{w^6\sigma_{\max}^2(s+1)^2}, \frac{C_6 a_1 \tau_{\min} n p \lambda}{w^3 \sigma_{\max}(s+1)} \right\} \right],\\
     &  \delta_2 \le 2 (s+1)\exp \left[-\min\left\{\frac{C_7 \tau_{\min}^2
    n p (\| \bbeta_{\mS}^{(0)}\|_{\min} - \gamma \lambda)^2 }{w^2\sigma_{\max}^2 (s+1)}, \frac{C_8 \tau_{\min} n p(\| \bbeta_{\mS}^{(0)}\|_{\min} - \gamma \lambda) }{w\sigma_{\max}(s+1)^{1/2}} \right\} \right],
\end{align*}
where $C_1,\dots,C_8$ are some positive constants, and $a_0 = \min\{1, a_2\}$. Moreover, $a_1, a_2, \gamma$ are constants specified in (i)--(iv).
In particular, if $np\lambda_0^2 / \{s\log(K)\} \to \infty$, then we have $\delta_0 + \delta_1 + \delta_2 \to 0$ as $np \to \infty$.
\eet

\subsection{Factor Composite Models}
\label{subsec:factor_scr}

Factor models, such as the capital asset pricing model (CAPM) and the Fama-French three-factor (FF3) model, have been widely used in the economics and finance \citep{perold2004capital,fama1992cross,fama1993common}.
By using a few effective factors, we can significantly reduce the number of parameters in large scale covariance matrix estimation \citep{fan2008high}.
In this subsection, we attempt to combine the classical factor models with our SCR model.
This leads to a new class of models, which combine the strengths from both the classical factor models and our SCR model. For convenience, we refer to this new class of methods as factor composite models.
Specifically, let $\by_i\in \mR^p \ (1\le i \le n)$ be the $n$ observations of the response vectors, and assume that $\bbf_i \in \mR^M \ (1\le i \le n)$ are the vectors of $M$ observable common factors. Then a typical factor model can be written as \citep{fan2008high}:
\begin{align}  \label{eq:factor_model}
    \by_i = \bB \bbf_i + \bu_i,
\end{align}
where $\bB= (\bb_1,\bb_2,\dots,\bb_M) \in \mR^{p\times M}$ is the unknown factor loading matrix, and $\bu_i\in \mR^p$ is the idiosyncratic error uncorrelated with the common factors. Without loss of generality, we assume that both $\bbf_i$ and $\bu_i$ have zero means. Then we should have $\bSigma= E(\by_i\by_i^\top)= \bB \bSigma_{\bbf} \bB^\top + \bSigma_{\bu}$, where $\bSigma_{\bbf} =E(\bbf_i\bbf_i^\top) \in
\mR^{M\times M}$ and $\bSigma_{\bu} = E(\bu_i\bu_i^\top)\in \mR^{p\times p}$. In a strict factor model, the covariance matrix $\bSigma_{\bu}$ of the idiosyncratic error is typically assumed to be diagonal \citep{fan2008high}. To enhance the model flexibility, we can model $\bSigma_{\bu}$ by our SCR model. That is $\bSigma_{\bu}(\bbeta) =  \sum_{k=0}^K \beta_k \bW_k$, where $\bW_k$s are the similarity matrices, and $\beta_k$s are the unknown coefficients.
Consequently, the covariance matrix $\bSigma $ is expressed as
\begin{align}\label{eq:factor_scr}
    \bSigma = \bB\bSigma_{\bbf}\bB^\top + \sum_{k=0}^K \beta_k \bW_k.
\end{align}

By model \eqref{eq:factor_scr}, an interesting finding arises when the factors are mutually uncorrelated, indicated by $\bSigma_{\bbf}=\diag\{\alpha_1^2,\dots,\alpha_M^2\}$ as a diagonal matrix.
This leads us to express model \eqref{eq:factor_scr} in a unified form as
\begin{align*}
    \bSigma = \sum_{m=1}^M \alpha_m^2 \bW_{\bb_m} + \sum_{k=0}^K \beta_k \bW_k,
\end{align*}
where $\bW_{\bb_m} = \bb_m \bb_m^\top\ (1\le m\le M)$ are rank-one matrices constructed based on the factor loadings.
There are several important differences between the two regression components. For example, consider the stock market. Note that the matrices $\bW_{\bb_m}$s are typically unobserved and need to be estimated using market-specific factors, such as those in the FF3 model. On the other hand, the similarity matrices $\bW_k$s can be directly observed or constructed using the collected firm-specific covariates $\bX_k$s from the financial statements of the firms. Furthermore, the summation of $\bW_{\bb_m}$s captures the low-rank factor structure of $\bSigma$, with the number of factors $M$ being relatively small or moderate. In contrast, the summation of $\bW_k$s captures a certain $\ell_1$-sparse structure of $\bSigma$, as the boundedness of $\|\bW_k\|_1$ is assumed in Condition (C\ref{cond:bounded_norm}).
It is worth noting that our approach also allows for a potentially large number of similarity matrices, specifically $K+1$, but only $s+1$ of them are actually useful.
In addition, the diagonal elements of $\bW_{\bb_m}$ can be distinct, which allows for modeling heterogeneous variance. On the other hand, the diagonal elements of matrix $\bW_k$ are usually the same, and in this case, we can model heterogeneous variance using the approach introduced in Section \ref{subsec:scr_n}. Lastly, while the elements of $\bW_{\bb_m}$s can be negative, similarity matrices $\bW_k$s often have non-negative elements.
Nevertheless, it is possible to construct similarity matrices with negative values using alternative approaches, as long as the regularity conditions as given before can be satisfied.
Inspired by an anonymous referee, we illustrate one possible approach by numerical studies in Section \ref{subsec:empirical} and Appendix A.7.

As we mentioned before, we refer to \eqref{eq:factor_scr} as a factor composite model.
To practically estimate the model \eqref{eq:factor_scr}, we adopt a similar procedures as suggested by \cite{fan2008high}. In the first step, we compute the least squares estimator of $\bB$ by $\wh\bB = (\bF^\top\bF)^\top \bF^\top \bY\in \mR^{M\times p}$, where $\bF = (\bbf_1^\top \dots,\bbf_n^\top)^\top \in \mR^{n\times M}$ and $\bY =(\by_1^\top \dots,\by_n^\top)^\top \in \mR^{n\times p} $.
Denote the residuals by $\wh\bu_i = \by_i - \wh \bB \bbf_i \in \mR^p$ for each $1\le i\le n$.
In the second step, we estimate the covariance of the residuals by the SCR method introduced in Section \ref{subsec:scr_n}. This yields the covariance matrix estimator $\wh \bSigma_{\bu} =\sum_{k=0}^K \wh\beta_k \bW_k$.
In the last step, we plug in all the components to obtain $\wh\bSigma = \wh\bB\wh\bSigma_{\bbf}\wh\bB^\top + \wh\bSigma_{\bu} $, where $\wh\bSigma_{\bbf} = n^{-1}\bF^\top\bF \in \mR^{M\times M}$ is the sample covariance matrix of the factors.
Numerical experiments as to be presented subsequently suggest that this factor model based SCR estimator works very well.

\section{Numerical Studies}

\subsection{Simulation Studies}
\label{subsec:simu}

\subsubsection{Simulation Settings and Algorithm Implementation}

In this section, we evaluate the finite sample performance of the folded concave penalized sparse covariance regression (SCR) method. The responses vector $\by$ is simulated by $\by = \bSigma_0^{1/2}\bZ$, where the components of the vector $\bZ$ are independently and identically generated from different distributions and will be specified later.
In addition, the true covariance matrix is set as $\bSigma_0 = \sum_{k=0}^{K} \bbeta_k^{(0)}\bW_k$, where $\bbeta^{(0)} = (\beta_0^{(0)}, \dots, \beta_K^{(0)})^\top = (8, 1, 1, 1, 0, \cdots, 0)^\top \in \mR^{K+1}$.
Then we have $\mS = \supp(\bbeta^{(0)}) = \{0, 1, 2, 3\}$ and $\mS^{c} = \{0, \dots, K\} \setminus \mS = \{4, \cdots, K\}$.
The off-diagonal elements of the similarity matrices $\bW_k = (w_{j_1j_2})\in\mR^{p\times p}, k = 1, \dots, K$ are independently and identically generated from Bernoulli distributions with probability $5p^{-1}$, and the diagonal elements are set to be 0.
We consider three different $(p,K)$ configurations, namely $(200,10), (500,100)$ and $(1000,1000)$ for the simulation.


For comparison, we consider both the SCAD penalty and the MCP penalty. We fix $\gamma = 3.7$  for the SCAD penalty as suggested by \cite{Fan:Li:2001}, and fix $\gamma = 1.5$ for the MCP penalty.
To choose an appropriate tuning parameter $\lambda$, we consider the following BIC-type criterion proposed in \cite{wang2009shrinkage}:
\beq
\BIC(\lambda) = \log\left( \left\|\by\by^\top - \sum_{k = 0}^K \wh\beta_k \bW_k \right\|_F^2\right) + \log\{\log(K+1)\}\frac{\log(p^2)}{p^2} \times df_\lambda,\label{eq:bic}
\eeq
where $df_\lambda$ is the number of nonzero coefficients in $\wh\bbeta=(\wh\beta_0,\dots,\wh\beta_K)^\top$.
Then we select $\lambda$ which minimizes the $\BIC(\lambda)$.
For the initial estimator in the LLA algorithm (i.e., Algorithm \ref{alg:LLA}), we use the Lasso estimator \eqref{eq:lasso} with the tuning parameter $\lambda_0$.
Our preliminary experiment showed that employing a single tuning parameter for both $\lambda_0$ and $\lambda$ yielded comparable results to selecting two separate tuning parameters. Therefore, to reduce computational costs, we set $\lambda_0=\lambda$ and select a single value for both $\lambda_0 $ and $\lambda$ using BIC.
Further details and discussion regarding this issue can be found in Appendix A.8. 
According to the discussion below \eqref{eq:lasso}, we do not penalize the intercept term $\beta_0$ in the numerical experiments.

\subsubsection{Performance Measurements and Simulation Results}

We then evaluate the sparse recovery and the estimation accuracy of the folded concave penalized SCR method. To obtain a reliable evaluation, the experiment is replicated for $R = 100$ times.
Let $\wh{\bbeta}^{(r)}$ be the estimated coefficients in the $r$th replication for $1\le r \le R$, and $\mS^{(r)} = \supp(\wh{\bbeta}^{(r)})$ be the corresponding set of indexes of nonzero estimated coefficients.
Then the covariance estimate in the $r$th replication can be written as  $\wh\bSigma^{(r)} = \bSigma(\wh{\bbeta}^{(r)}) = \sum_{k = 0}^K\wh{\beta}_k^{(r)} \bW_k$.
We first investigate the sparse recovery property of the folded concave penalized SCR method.
In this regard, we consider three measurements.
The first one is the true positive rate ($\TPR$), defined by $\TPR = R^{-1}\sum_{r=1}^R | \mS^{(r)}\cap\mS |/ |\mS|$.
The second one is the false positive value ($\FPR$), defined by $\FPR = R^{-1}\sum_{r=1}^R | \mS^{(r)}\setminus \mS| / |\mS^{(r)}|$.
We also report the fraction of corrected selection defined by $\CS = R^{-1}\sum_{r=1}^R I\{\mS^{(r)} = \mS\}$, where $I\{\cdot\}$ is the indicator function.
Next, we evaluate the estimation accuracy. To this end, we calculate the root mean squared error ($\rmse$) for the coefficient $\bbeta$ as $\rmse_{\bbeta} = \sqrt{(RK)^{-1}\sum_{k=0}^K\sum_{r=1}^R(\wh{\beta}_k^{(r)} - \beta_k^{(0)})^2}$, bias ($\Bias$) and the standard deviation ($\SD$) for the coefficient $\bbeta$ as
$\Bias_{\bbeta} = K^{-1}\sum_{k=0}^K
 \vert\bar\beta_k - \beta_k^{(0)}\vert$ and $\SD_{\bbeta} = \sqrt{(RK)^{-1}\sum_{k=0}^K\sum_{r=1}^R(\wh{\beta}_k^{(r)} - \bar\beta_k)^2}$, with $\bar\beta_k= R^{-1}\sum_{r=1}^R\wh{\beta}_k^{(r)}, 0\le k \le K$, respectively.
Lastly, we evaluate the performance of the estimated covariance matrices. Following \cite{zou2017covariance}, we consider the spectral error and the Frobenius error of the estimated covariance matrices measured under the spectral norm and the Frobenius norm, i.e., $R^{-1}\sum_{r=1}^R\|\wh\bSigma^{(r)}-\bSigma_{0}\|_{2}$ and $R^{-1}\sum_{r=1}^Rp^{-1 / 2}\|\wh\bSigma^{(r)}-\bSigma_{0}\|_{F}$.
For comparison, we also compute the corresponding performance measurements for the OLS estimator \eqref{eq:ols_est} and the oracle estimator \eqref{eq:oracle}.

\begin{table}[htbp]
    \setlength{\abovecaptionskip}{1.0cm}
    \setlength{\belowcaptionskip}{0.3cm}
	\begin{center}
		\renewcommand\tablename{Table}
		\caption{Simulation results for $\bZ$ generated from the standard normal distribution.}
		\label{tab:sim_sn}
  \resizebox{\textwidth}{!}{
		\begin{tabular}{cc|ccc|ccc|cc}
			\hline
$(p,K)$&Penalty&$\TPR$&$\FPR$&$\CS$&$\rmse$&Bias&$\SD$&$\|\cdot\|_2$&$\|\cdot\|_F$\\
\hline
\multirow{4}{*}{(200,10)}
    &SCAD    &  0.800 &  0.091 &  0.190 & 0.471 & 0.051 & 0.465 &  8.026 &  2.732 \\
    &MCP     &  0.795 &  0.091 &  0.170 & 0.473 & 0.052 & 0.467 &  8.095 &  2.754 \\
    &OLS     & --     & --     & --     & 0.480 & 0.032 & 0.479 &  8.596 &  2.898 \\
    &ORACLE  &  1.000 &  0.000 &  1.000 & 0.363 & 0.016 & 0.361 &  4.902 &  1.731 \\
\hline
\multirow{4}{*}{(500,100)}
    &SCAD    &  0.940 &  0.049 &  0.580 & 0.090 & 0.005 & 0.087 &  4.582 &  1.524 \\
    &MCP     &  0.940 &  0.049 &  0.580 & 0.090 & 0.005 & 0.087 &  4.583 &  1.524 \\
    &OLS     & --     & --     & --     & 0.229 & 0.018 & 0.228 & 16.240 &  5.048 \\
    &ORACLE  &  1.000 &  0.000 &  1.000 & 0.067 & 0.002 & 0.065 &  2.921 &  1.011 \\
\hline
    \multirow{4}{*}{(1000,1000)}
    &SCAD    &  0.990 &  0.046 &  0.770 & 0.021 & 0.000 & 0.021 &  3.263 &  0.991 \\
    &MCP     &  0.990 &  0.048 &  0.760 & 0.021 & 0.000 & 0.021 &  3.324 &  1.003 \\
    &OLS     & --     & --     & --     & 0.160 & 0.013 & 0.159 & 30.888 & 11.282 \\
    &ORACLE  &  1.000 &  0.000 &  1.000 & 0.016 & 0.000 & 0.015 &  2.095 &  0.723 \\
\hline
		\end{tabular}}
	\end{center}
\end{table}

We consider that the components of $\bZ$ are independently and identically generated from (i) a standard normal distribution $\mN(0,1)$, (ii) a mixture normal distribution $\xi \cdot \mN(0, 5/9) + (1-\xi)\cdot \mN(0, 5)$ with $P(\xi=1)=0.9$ and $P(\xi=0)=0.1$, or (iii) a standardized exponential distribution Exp$(1)-1$.
The simulation results for the standard normal distribution are given in Table \ref{tab:sim_sn}.
Since all three distributions present similar results, to save space, we relegate the simulation results of the mixture normal and the standardized exponential distributions to the supplementary material; see Tables A.1--A.2 in Appendix A.7. 
We next focus on Table \ref{tab:sim_sn}.
Considering sparsity recovery, it can be observed that as $p$ increases, the $\TPR$ values of both SCAD and MCP estimators gradually increase, while the $\FPR$ values decrease. In addition, the proportion of correct selection of all non-zero coefficients also gradually increases.
This verifies the selection consistency of the proposed method and demonstrates the usefulness of the BIC criterion.
Regarding the accuracy of the coefficient estimation, we can see that the $\rmse$, $\Bias$, and $\SD$ values of all the estimators decrease as $p$ increases.
However, the $\rmse$ and $\SD$ values for the OLS estimator are much higher compared to the other three estimators, especially when both $p$ and $K$ are large.
In contrast, as $p$ increases, the estimation errors of SCAD and MCP estimators gradually approach those of the optimal oracle estimator.
This observation confirms the oracle property for the two penalized estimators obtained through the LLA algorithm.
Lastly, in terms of the estimation of the covariance matrix, we can see that as $p$ increases, both two error measurements of the two penalized estimators get close to those of the oracle estimator. In contrast, the estimation errors of the OLS estimator increase with the growth of both $p$ and $K$. This finding suggests that the covariance matrix obtained by the OLS method is inconsistent when the number of predictors $K$ diverges too fast.
All these results demonstrate the effectiveness of the folded concave penalized estimation for the SCR model.

\subsection{A Case Study with Stocks of Chinese A-Share Market}
\label{subsec:empirical}

In this subsection, we apply the proposed sparse covariance regression (SCR) model to analyze the returns of the stocks traded in the Chinese A-Share market.
We first describe the data and covariates used to construct the similarity matrix.
Subsequently, we employ the SCR method to select the similarity matrices for the corresponding covariance matrix estimation.
This allows us to construct a portfolio with the estimated covariance matrix.
We then evaluate the portfolio's investment performance and illustrate the proposed methodology's usefulness.

\subsubsection{Data Description}
\label{subsec:data_desc}

In this study, we collect quarterly returns of  $p = 667$ stocks of the Chinese A-share market
after the basic data cleaning procedure.
Specifically, the stocks are obtained with complete return and covariate information
during the year 2016 to 2020.
It
leads to a total of $T = 20$ quarters.
The stock information is collected from the Chinese Stock Market and Accounting Research (CSMAR) database (https://us.gtadata.com/csmar.html).
We first present some descriptive data analysis as follows.
First, for each stock $j$, we calculate the average return of the stock
as $T^{-1}\sum_t Y_{jt}$.
Then it yields the histogram in the left panel of Figure \ref{ts_hist}.
We can obtain that the average returns of stocks range from -0.1 to 0.2, with the majority lying between -0.05 and 0.05.
In addition, we calculate the average stock return for each time point
as $p^{-1}\sum_j Y_{jt}$, leading to the time series in the right panel of Figure \ref{ts_hist}. The average stock returns have the lowest level in the first quarter and reach their highest in the 13th quarter (i.e., the first quarter of 2019).
Indicated by the existing theoretical and empirical studies (e.g., \cite{ROLL1988R2} and \cite{zou2017covariance}), the stock return comovement can be closely related to the firm's fundamentals.
We are then motivated to consider several firms' fundamentals for constructing the similarity matrices in the covariance regression model.
Specifically, we collect $11$ covariates from the financial statements of the firms, including the
SIZE (logarithm of market value),
BM (book-to-market ratio),
CR (cash ratio of the firm, measuring the liquidity of the firm),
WARE (weighted return on equity),
OER (owner's equity ratio, measuring the firm's long-term solvency),
TAT (total asset turnover, measuring the firm's operational efficiency of assets),
RTA (return on total assets),
CF (cash flow of the firm),
LEV (leverage ratio),
CAAR (capital accumulation rate, measuring the firm's development ability),
and EPS (earning per share).
These covariates provide measurements of the firms'
performances in various aspects \citep{bodie2020investments,palepu2020business}.
Lastly, all covariates are standardized with mean 0 and variance 1.

\begin{figure}[htbp]
    \centering
    \subfigure{
        \begin{minipage}[t]{0.5\linewidth}
            \centering
            \includegraphics[scale=0.45]{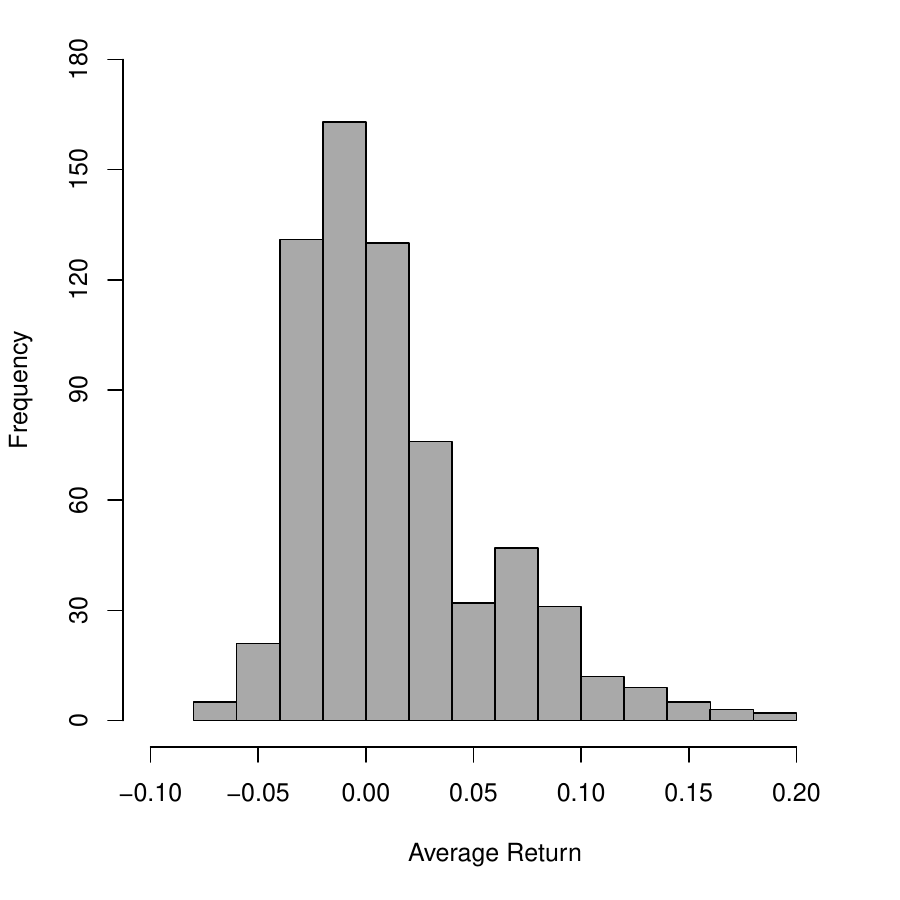}
        \end{minipage}%
    }%
    \subfigure{
        \begin{minipage}[t]{0.5\linewidth}
            \centering
            \includegraphics[scale=0.45]{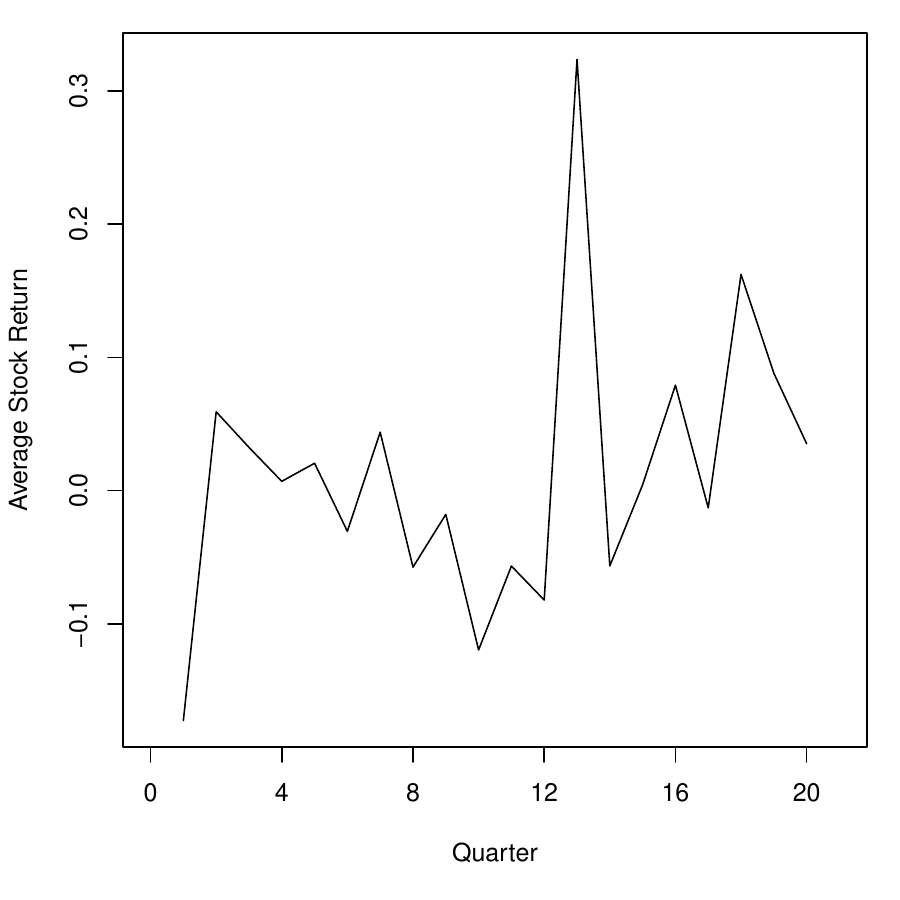}
        \end{minipage}
    }
    \centering
    \caption{The left panel: histogram of the average return of $p = 667$ stocks; The right panel: the time series of average stock returns over $T = 20$ quarters.}\label{ts_hist}
\end{figure}

Subsequently, we construct the similarity matrices as follows.
First, for the $k$th covariate $\bX_k = (X_{1k},\cdots, X_{pk})^\top \in \mR^p$, we construct the associated similarity matrices $\bW_k = (w_{k,j_1j_2}) \in \mR^{p\times p}$ using two different approaches.
Specifically, for the first approach, we define $ w_{k,j_1j_2} = \exp\{-10(X_{j_1k} - X_{j_2k})^2\}$ if $(X_{j_1k} - X_{j_2k})^2<\tau_k$, and $w_{k,j_1j_2}=0$ if $(X_{j_1k} - X_{j_2k})^2>\tau_k$ or $j_1=j_2$. Here, we choose $\tau_k>0$ such that each $\bW_k$ has $1/4$ nonzero elements.
For the second approach, we  define $ \bW_k = \bX_k \bX_k^\top / p $.
Then for the $11$ covariates, we can construct a total of $22$ similarity matrices.
Subsequently, we construct two additional similarity matrices based on the stock industrial network and common shareholder network.
For the stock industrial network, (denoted as $\bW_{\ind} = (w_{\ind,j_1j_2})$), we define $w_{\ind,j_1j_2} = 1$ if the stock $j_1$ and stock $j_2$ belong to the same industry, otherwise $w_{\ind,j_1j_2} = 0$.
Here, all stocks are categorized into 14 industries according to the China Securities Regulatory Commission (2012 edition).
In addition, we denote the common shareholder network as $\bW_{\sh} = (w_{\sh, j_1j_2})$, where $w_{\sh, j_1j_2} = 1$ if the stock $j_1$ and stock $j_2$ share at least one top ten shareholders, otherwise $w_{\sh, j_1j_2} = 0$. This leads to a total of $K = 24$ similarity matrices $\bW_k$ $(1\le k\le K)$. Lastly, we rescale the elements of similarity matrices so that $\|\bW_k\|_1 = 1$ for each $1\le k\le K$.

\begin{figure}[ht]
    \centering
    \subfigure{
        \begin{minipage}[t]{0.5\linewidth}
            \centering
            \includegraphics[scale=0.45]{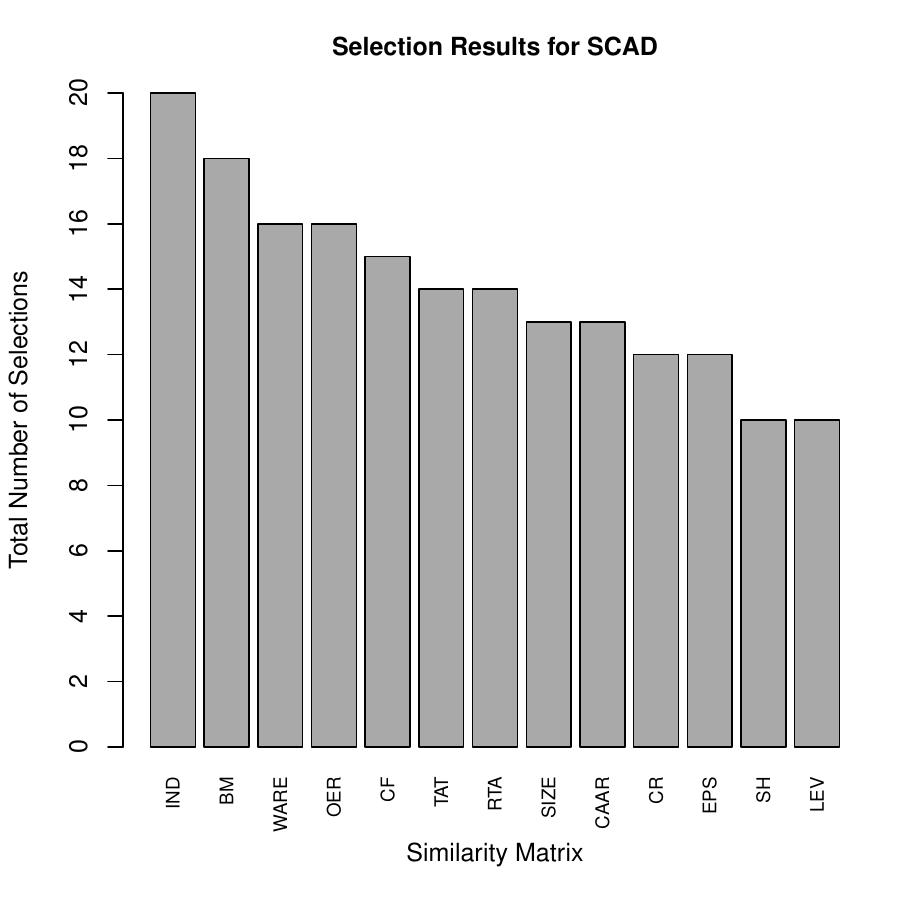}
        \end{minipage}%
    }%
    \subfigure{
        \begin{minipage}[t]{0.5\linewidth}
            \centering
            \includegraphics[scale=0.45]{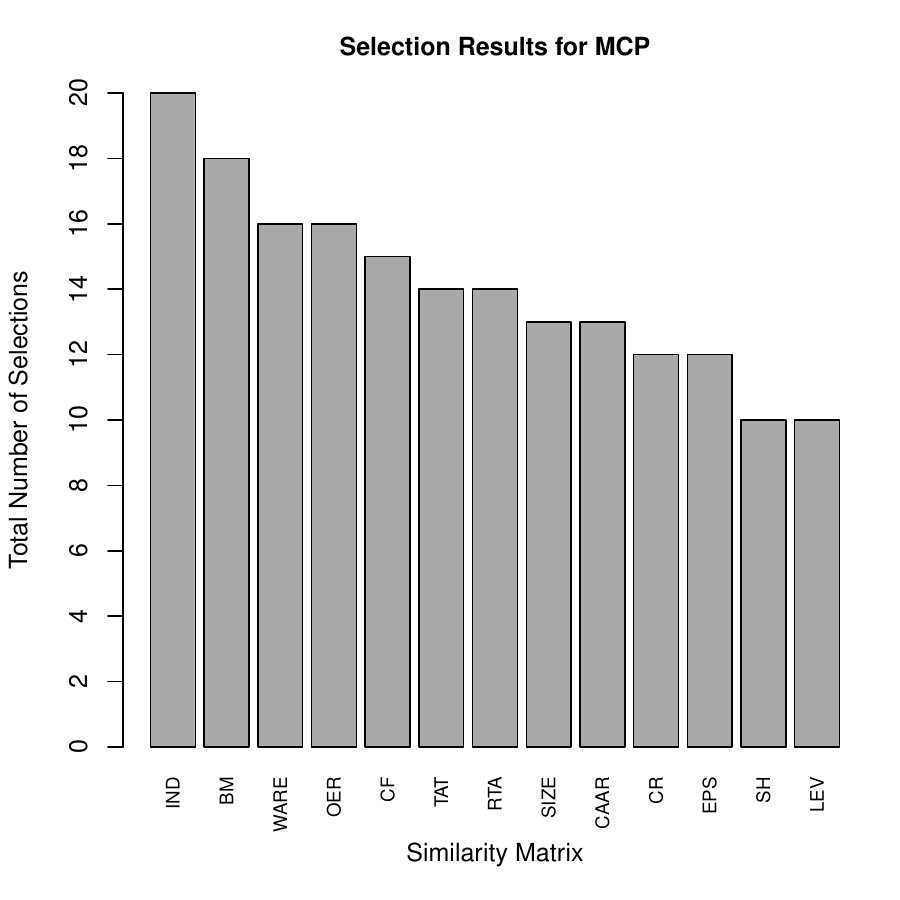}
        \end{minipage}
    }
    \centering
    \caption{The left panel: the total number of selections for each similarity matrix during all 20 fittings using the SCAD penalty; The right panel: the total number of selections for each similarity matrix during all 20 fittings using the MCP penalty.}\label{sim_sel}
\end{figure}

\subsubsection{Model Estimation and Evaluation}

Subsequently, we apply the SCR model with SCAD and MCP penalties to the stock return data.
We adopt a rolling window approach for model training and evaluation.
Specifically, we set $n = 1$ as the training window size and fit the model for $T = 20$ times.
We also calculate the total number of selections for these similarity matrices. Note that for the similarity matrices constructed from the same covariate, we only count them once.
The results are shown by bar plots in Figure \ref{sim_sel}.
Here, the left panel corresponds to the SCAD penalty, and the right panel corresponds to the MCP penalty. Both penalties yield nearly identical selection results.
In summary, IND, BM, WARE and OER are the top four most frequently selected matrices for both the SCAD penalty and the MCP penalty.
It reflects their importance in this covariance regression modeling problem.

Then we utilize the covariance regression result for the portfolio
construction and investment. After we obtain the fitted covariance matrix, to ensure its positive-definiteness, we set its non-positive eigenvalues to be $\epsilon = 10^{-6}$ and keep the eigenvectors unchanged.
Suppose the estimated covariance at the $t$th quarter is $\wh \bSigma_t$.
To construct the optimal portfolio, we solve the global minimal variance portfolio problem as
$\bomega_{t}^{*}=\arg \min _{\bomega^\top \one = 1} \bomega^{\top} \wh{\bSigma}_t \bomega$,
where $\bomega = (\omega_1,\cdots, \omega_p)^\top \in \mR^p$.
Then we assess the portfolio return in the subsequent quarter
by $ \bomega_t^{*\top} \by_{t+1}$.
For model comparison, we first calculate the market portfolio as a benchmark, which is
the average of all stock returns in the next quarter with weights proportional to their market capitalization.
Furthermore, we include the unpenalized OLS estimator \eqref{eq:ols_est} for the covariance regression model, including all the similarity matrices.

We examine the portfolio performance by five commonly used measures (e.g., see \cite  {bodie2020investments}).
They are,
Mean (the average return of investment portfolios);
SD (the standard deviation of the portfolio returns over the investing period, interpreted as the risk of the portfolio);
Sharpe ratio (excess return over the risk-free rate adjusted by SD); Alpha (the alpha coefficient is a
the risk-adjusted excess return of the investment portfolio over the benchmark);
Beta (the beta coefficient close to 1 indicates the out-of-sample portfolio has almost the same volatility as the benchmark).
Besides,  we further present the compound quarterly growth rate (CQGR) of the four portfolios,
which is calculated by $\big\{\prod_{t=2}^T (1 + r_{t})\big\}^{1/(T-1)} - 1$
and $r_{t}$ is the return of the $t$th quarter.

\begin{table}[htbp]
	{
	\setlength{\abovecaptionskip}{1.0cm}
	\setlength{\belowcaptionskip}{0.3cm}
	\begin{center}
		\caption{The quarterly Mean, SD, Sharpe ratio, Alpha, Beta, and compound quarterly growth rate (CQGR) of the two penalized, the unpenalized OLS, and the market portfolio returns (\%).}
		\label{inv}
		\begin{tabular}{ccccccc}
			\hline
			& Mean& SD& Sharpe Ratio& Alpha& Beta & CQGR\\
			\hline
			SCAD	& 4.206 & 10.647 & 0.360 & 1.869 & 0.803 & 3.717 \\	
			MCP		& 4.206 & 10.647 & 0.360 & 1.869 & 0.803 & 3.717 \\	
			OLS  	& 2.248 & 9.431 & 0.199 & -0.614 & 0.983 & 1.857 \\	
			Market	& 2.913 & 8.197 & 0.310 & 0.000 & 1.000 & 2.612 \\
			\hline
		\end{tabular}
	\end{center}
}
\end{table}

Table \ref{inv} presents the constructed four portfolios on the above measures.
We can observe that for both the SCAD penalty and the MCP penalty, the penalized portfolios have higher mean returns compared to the unpenalized OLS and the market portfolios, although their standard deviations are moderately higher than the market.
After adjusting for the risks, the two portfolios still have higher Sharpe ratios and alpha coefficients than the other competing methods, and their Beta coefficients are also smaller than one.
In particular, the two penalized portfolios have the CQGR of $3.717\%$, which is higher than the other two methods.
In summary, the above investment results demonstrate the superiority of the constructed portfolios with our proposed SCR method.

\subsubsection{Daily Return Data}
To further demonstrate the usefulness of the SCR model, we compare our method with some popularly used methods on daily stock returns data.
Specifically, we collected the daily returns for the same 667 stocks mentioned earlier, spanning 20 quarters from 2016 to 2020.
After data cleaning, a total of $p=283$ daily stock returns for $1218$ trading days are retained.
To apply the capital asset pricing model (CAPM) and the Fama-French three-factor (FF3) model, we also collect three common factors for each trading day from the RESSET financial research database (http://www.resset.cn/endatabases). They are, respectively, the market factor (\texttt{MKT}), the size factor (\texttt{SMB}), and the value factor (\texttt{HML}). We also construct the $K=24$ similarity matrices $\bW_k\in\mR^{p\times p}$ for the $p=283$ stocks as in the above subsection.

Then we adopt the rolling window approach for model training and evaluation. Specifically, at the first day of each quarter, we use the daily return data of the preceding one quarters (i.e., $n\approx 60$) as the training dataset to construct portfolios by different methods. We consider the following covariance matrix estimation methods. The first one is our SCR method for repeated responses as introduced in Section \ref{subsec:scr_n}. Since the two folded concave penalties have shown similar performance, we will only use the SCAD penalty for the SCR method.
We also consider two strict factor models to estimate the covariance matrix. The first one is the CAPM with the single market factor \texttt{MKT}. The second one is the FF3 model with all three factors \texttt{MKT}, \texttt{SMB}, and \texttt{HML}. In addition, the factor composite models as discussed in Section \ref{subsec:factor_scr} are also examined.
Another way to implement the factor model \eqref{eq:factor_model} is to treat the 11 covariates described in Section \ref{subsec:data_desc} as known factor loadings.
Then we run the cross-sectional regression on these loadings to obtain the factors and residuals.
The residual covariance can be estimated by two different methods.
The first one is to estimate the covariance of the residuals by a diagonal matrix, similar to the strict factor model. The second one is to use our SCR model with $K=24$ similarity matrices to estimate the covariance of the residuals. Finally, we obtained the complete covariance matrix of returns by adding the covariance of the factor part and the residual part. The two models are referred to as  characteristics-based factor (CBF) model and ``CBF + SCR'' model respectively.
Lastly, we consider the shrinkage method of \cite{ledoit2004well}, which will be referred to as the LW method. According to their approach, the covariance matrix can be estimated by $\wh \bSigma_\textup{LW} = \rho\{\tr(\wh\bS)/p\}\bI_p + (1-\rho) \wh\bS$, where $\wh \bS$ is the sample covariance matrix of the daily returns, and $\rho\in[0,1]$ can be calculated as in Section 3.3 of \cite{ledoit2004well}.
By replacing $\wh\bS$ with our SCR estimator, another composite estimator can be obtained.
After obtaining the covariance estimator $\wh{\bSigma}$, we then solve $\bomega^{*}=\arg \min _{\bomega^\top \one = 1} \bomega^{\top} \wh{\bSigma} \bomega$ to construct the portfolio.
Then we assess each portfolio return in the subsequent quarter. This leads to a total of $19$ quarterly investment returns for each portfolio.
The Mean, SD, and Sharpe ratio for the quarterly returns of each portfolio are presented in Table \ref{tab:inv2}. For comparison, we also calculate the market portfolio as a benchmark.

\begin{table}[htbp]
	{
	\setlength{\abovecaptionskip}{1.0cm}
	\setlength{\belowcaptionskip}{0.3cm}
		\begin{center}
			\caption{The Mean, SD, and Sharpe ratio of the quarterly returns for different portfolios (\%).}
			 \label{tab:inv2}
			 \resizebox{\textwidth}{!}{
	
			 \begin{tabular}{c|cccccc|cccc}
				\hline
				&\multicolumn{6}{c|}{Individual Methods} & \multicolumn{4}{c}{Composite Methods}\\
				 &Market & CAPM  & FF3 & CBF & LW  & SCR &  CAPM+SCR & FF3+SCR  & CBF+SCR & LW+SCR \\
				 \hline
				 Mean            & 3.029 & 1.646 & 1.694 & 2.390 & 2.377 & 3.940 & 3.001 & 2.963 & 3.494 & 3.596 \\
				 SD              & 7.555 & 5.431 & 5.022 & 8.707 & 5.327 & 8.819 & 5.345 & 5.022 & 7.541 & 7.414 \\
				 Sharpe Ratio    & 0.352 & 0.234 & 0.263 & 0.232 & 0.376 & 0.404 & 0.492 & 0.516 & 0.414 & 0.435 \\
	
			\hline
			\end{tabular}}
		\end{center}
	}
	\end{table}

From Table \ref{tab:inv2}, we can obtain the following observations.
First, for each individual method, it can be observed that the three strict factor models (i.e., CAPM, FF3 and CBF) have comparable performance, but their Sharpe ratios are much lower than that of the Market. In addition, the SCR and LW methods have better performance than the Market in terms of Sharpe ratio. Furthermore, for these composite methods, it is evident that all the four composite models (i.e., CAPM+SCR, FF3+SCR,  CBF+SCR, and LW+SCR) show a great improvement in Sharpe ratio as compared with their non-composite counterparts. In particular, the combination of FF3 and SCR method yields the highest Sharpe ratio $0.516$.

\section{Conclusion}

This work investigates the penalized estimation of the sparse covariance regression (SCR) model.
Specifically, we first examine the Lasso estimator and derive its non-asymptotic error bound.
Subsequently, we compute the folded concave penalized estimator using the local linear approximation (LLA) algorithm, with the Lasso estimator as the initial value.
Theoretical analysis demonstrates that the resulting estimator can converge to the oracle estimator with overwhelming probability under appropriate regularity conditions.
Additionally, we establish the asymptotic normality of the oracle estimator under more general conditions.
We also extend the SCR method to the scenarios with repeated observations of the response.
Finally, we demonstrate the usefulness of the proposed method on a Chinese stock market dataset.

We briefly discuss possible future research directions.
Firstly, we provide a criterion to select the tuning parameters from the application point of view. It is also meaningful to investigate its theoretical performance rigorously.
Secondly, when dimension $p$ is very large, the computational burden of the SCR model becomes a crucial issue. Therefore, it is of great interest to design more computationally efficient methods.
Lastly, it is known that quantile regression is more robust to heavy-tailed noise than the ordinary least squares regression. Therefore, replacing the current quadratic loss with a check loss should also be a challenging but valuable extension.

\section*{Acknowledgment}

The authors are very grateful to the Editor, Associate Editor, and two anonymous reviewers for their constructive comments that greatly improved the quality of this paper. 
Yuan Gao's research is supported by the Postdoctoral Fellowship Program of CPSF (GZC20230111) and the National Natural Science Foundation of China  (No. 72471254). 
Xuening Zhu's research is supported by the National Natural Science Foundation of China (nos. 72222009, 71991472, 12331009), Shanghai International Science and Technology Partnership Project (No. 21230780200), Shanghai B\&R Joint Laboratory Project (No. 22230750300), MOE Laboratory for National Development and Intelligent Governance, Fudan University, IRDR ICoE on Risk Interconnectivity and Governance on Weather/Climate Extremes Impact and Public Health, Fudan University. 
Tao Zou’s research is supported by the ANU College of Business and Economics Early Career Researcher Grant, and the RSFAS Cross Disciplinary Grant. 
Hansheng Wang's research is partially supported by the National Natural Science Foundation of China (No. 12271012).

\section*{Disclosure Statement}
The author reports there are no competing interests to declare.

\bibliographystyle{asa}
\bibliography{xuening}

\newpage

\iftrue  
{
	
	\appendix
	
	\setcounter{table}{0}
	\renewcommand{\thetable}{A.\arabic{table}}
	
	\setcounter{equation}{0}
	\renewcommand{\theequation}{A.\arabic{equation}}
	
	\section{Appendix}

	\subsection{Proof of Theorem \ref{thm:lasso_est}}
	\label{proof of thm:lasso_est}
	
	\begin{proof}
		
		We follow the proof idea of Theorem 7.13 (a) in \cite{wainwright2019high}.
		Recall that $\by\by^\top = \sum_{k=0}^K \beta_k^{(0)} \bW_k + \mE$.
		Define $\wh \bdelta \defeq \wh\bbeta^{\lasso} - \bbeta^{(0)}$.
		We first show that, if $\lambda_0 \ge (2/p) \max_{0\le k\le K} |\tr(\bW_k\mE)|$ holds, then $\wh\bdelta \in \C_3(\mS) \defeq \{\bdelta \in \mR^{K+1}: \|\bdelta_{\mS^c}\|_1 \le 3 \|\bdelta_{\mS}\|_1\}$.
		Subsequently, we show that $\big\{\lambda_0 \ge (2/p) \max_{k\in\mS} |\tr(\bW_k\mE)|\big\}$ holds with high probability.
		
		\noindent
		{\bf Step 1.}
		Since $\wh \bbeta^{\lasso}$ is the solution to the problem \eqref{eq:lasso}, we have
		\begin{equation*}
			Q(\wh\bbeta^{\lasso}) +\lambda_0  \| \wh \bbeta^{\lasso} \|_1 =  \frac{1}{2 p}\left\|\mE - \sum_{k = 0}^K \wh\delta_k \bW_k \right\|_F^2 + \lambda_0  \| \wh \bbeta^{\lasso} \|_1 \le \frac{1}{2 p}\|\mE \|_F^2 + \lambda_0  \|\bbeta^{(0)} \|_1.
		\end{equation*}
		Rearranging the above inequality, we obtain that
		\begin{equation}\label{ineq:lagrangian}
			0 \le \frac{1}{2p} \left\| \sum_{k=0}^K \wh\delta_k \bW_k \right\|_F^2 \le \frac{1}{p} \tr\left( \mE \sum_{k=0}^K \wh\delta_k \bW_k \right) + \lambda_0 \Big\{\|\bbeta^{(0)}\|_1 - \|\wh\bbeta^{\lasso}\|_1 \Big\}
		\end{equation}
		Note that
		\begin{equation}\label{ineq:holder}
			\tr\left( \mE \sum_{k=0}^K \wh\delta_k \bW_k \right) \le \sum_{k=0}^K |\wh\delta_k | \cdot |\tr\left( \bW_k\mE \right)| \le \|\wh\bdelta\|_1 \max_{0\le k\le K} |\tr(\bW_k\mE)|.
		\end{equation}
		Since $\bbeta^{(0)}$ is supported on $\mS$, we can write $    \|\bbeta^{(0)}\|_1 - \|\wh\bbeta^{\lasso}\|_1 = \|\bbeta_{\mS}^{(0)}\|_1 - \|\bbeta_{\mS}^{(0)} + \wh\bdelta_{\mS}\|_1 - \|\wh\bdelta_{\mS^c}\|_1$.
		Substituting it into the inequality \eqref{ineq:lagrangian} and using the inequality \eqref{ineq:holder} yields
		\begin{align}
			0 \le& \frac{1}{p} \left\| \sum_{k=0}^K \wh\delta_k \bW_k \right\|_F^2 \le \frac{2}{p} \max_{0\le k\le K} |\tr(\bW_k\mE)| \cdot \|\wh\bdelta\|_1+ 2\lambda_0 \Big\{\|\bbeta_{\mS}^{(0)}\|_1 - \|\bbeta_{\mS}^{(0)} + \wh\bdelta_{\mS}\|_1 - \|\wh\bdelta_{\mS^c}\|_1\Big\} \nonumber\\
			\le & \lambda_0 \|\wh\bdelta\|_1+ 2\lambda_0 \Big\{ \|\wh\bdelta_{\mS}\|_1 - \|\wh\bdelta_{\mS^c}\|_1\Big\} \le \lambda_0 \Big\{3 \|\wh\bdelta_{\mS}\|_1 - \|\wh\bdelta_{\mS^c}\|_1 \Big\}, \label{ineq:bound}
		\end{align}
		where we have used the condition $\lambda_0 \ge  (2/p) \max_{0\le k\le K} |\tr(\bW_k\mE)|$ in the third inequality.
		Thus, we conclude that $\wh\bdelta \in \C_3(\mS)$.
		Then, by the RE Condition (C\ref{cond:RE}) and the inequality \eqref{ineq:bound}, we can obtain that
		\begin{align*}
			\kappa \|\wh\bdelta\|^2 \le \frac{1}{p} \left\| \sum_{k=0}^K \wh\delta_k \bW_k \right\|_F^2 \le \lambda_0 \Big\{3 \|\wh\bdelta_{\mS}\|_1 - \|\wh \bdelta_{\mS^c}\|_1 \Big\} \le 3 \lambda_0 \sqrt{s+1} \| \wh\bdelta\|,
		\end{align*}
		where the last inequality follows from \eqref{ineq:vector_21} in Lemma \ref{lemma:norm_ineq} with $\|\wh\bdelta_{\mS}\|_1 \le \sqrt{s+1} \|\wh\bdelta_{\mS}\| \le \sqrt{s+1} \|\wh\bdelta\|$.
		This implies the conclusion $\|\wh\bbeta^{\lasso} - \bbeta^{(0)} \| = \|\wh\bdelta\| \le (3/\kappa)\sqrt{s+1} \lambda_0 $.
		
		\noindent
		{\bf Step 2.}
		It remains to show that the event $\big\{\lambda_0 \ge  (2/p) \max_{0\le k\le K} |\tr(\bW_k\mE)| \big\}$ holds with high probability.
		Recall that $\tr(\bW_k\mE)=\by^\top\bW_k\by - \tr(\bW_k\bSigma_0) $.
		Further note that Condition (C\ref{cond:bounded_norm}) and norm inequality \eqref{ineq:matrix_21} in Lemma \ref{lemma:norm_ineq} imply that $\sup_{p,k}\|\bW_k\| \le \sup_{p,k}\|\bW_k\|_1 \le w $ and $\|\bSigma_0\| \le \|\bSigma_0^{1/2}\|^2 \le \|\bSigma_0^{1/2}\|_1^2 \le \sigma_{\max}$.
		Then by union bound and Lemma \ref{lemma:hanson_wright}, we have
		\begin{align*}
			P\left\{ \frac{2}{p} \max_{0\le k\le K} |\tr(\bW_k\mE)| \ge \lambda_0 \right\} \le& \sum_{k=0}^K P\left( \big|\by^\top\bW_k\by - \tr(\bW_k\bSigma_0) \big| \ge \frac{p\lambda_0}{2}\right) \\
			\le& 2(K+1)\exp\left\{ -\min\left( \frac{C_1 p \lambda_0^2}{w^2\sigma_{\max}^2},\frac{C_2 p \lambda_0}{w\sigma_{\max}} \right)\right\}.
		\end{align*}
		Thus, we should have the event $\big\{\lambda_0 \ge  (2/p) \max_{0\le k\le K} |\tr(\bW_k\mE)| \big\}$ holds with the probability at least $1 - 2(K+1)\exp\left\{ -\min\left( \frac{C_1 p \lambda_0^2}{w^2\sigma_{\max}^2},\frac{C_2 p \lambda_0}{w\sigma_{\max}} \right)\right\}$. This completes the proof of the theorem.
	\end{proof}

	\noindent\textbf{Remark.}
	In Theorem \ref{thm:lasso_est}, we establish the $\ell_2$-bound for the lasso estimator $\wh\bbeta^{\lasso}$.
	In the subsequent analysis for the LLA algorithm, this $\ell_2$-bound is used to obtain the $\ell_\infty$-bound $\| \wh\bbeta^{\lasso} - \bbeta^{(0)}\|_{\infty}$ by applying the norm inequality \eqref{ineq:vector_max2} in Lemma \ref{lemma:norm_ineq}.
	This will lead to an extra factor $\sqrt{s}$ between the two tuning parameters $\lambda_0$ and $\lambda$.
	In fact, we may get rid of the factor $\sqrt{s}$ by directly establishing the $\ell_\infty$-bound of the Lasso estimator.
	Then we can relax the the requirement of $\lambda$ in Theorem \ref{thm:convergence} to be $\lambda \ge c \lambda_0$ for some constant $c>0$.
	This can be done by replacing the restricted eigenvalue (RE) Condition (C\ref{cond:RE}) with a restricted invertibility factor (RIF) type condition \citep{zhang2012general}:
	\begin{enumerate} [(C1')]
		\setcounter{enumi}{4}
		\item (\textsc{Restricted Invertibility Factor}) Define the set $\C_3(\mS) \defeq \{\bdelta \in \mR^{K+1}: \|\bdelta_{\mS^c}\|_1 \le 3 \|\bdelta_{\mS}\|_1\}$.
		Assume $\{\bW_k\}_{0\le k \le K}$ satisfies the restricted invertibility factor (RIF) condition, that is,
		\begin{equation*}
			\frac{1}{p} \left\|\bSigma_{W} \bdelta \right\|_{\infty} \ge  \kappa' \| \bdelta\|_\infty, \quad \text{for all $\bdelta \in \C_3(\mS)$}
		\end{equation*}
		for some constant $\kappa'>0$, where $\bSigma_{W} = \{\tr(\bW_k\bW_l):0\le k,l\le K \} \in \mR^{(K+1)\times (K+1)}$.\label{cond:RIF}
	\end{enumerate}
	
	We next use Condition (C\ref{cond:RIF}') to establish the $\ell_{\infty}$-bound. By \eqref{ineq:bound} in the proof of Theorem \ref{thm:lasso_est}, we know that $\wh\bdelta = \wh\bbeta^{\lasso}-\bbeta^{(0)}\in \C_3(\mS)$. Thus, RIF condition implies that $\|\wh\bdelta\|_\infty \le \|\bSigma_{W} \wh\bdelta \|_{\infty} / (p\kappa')$.
	Note that
	\begin{align*}
		\bSigma_{W} \wh\bdelta = \bSigma_{W}(\wh\bbeta^{\lasso}-\bbeta^{(0)}) = \tr\left\{\bW_k \left(\sum_{l=0}^K \wh\beta^{\lasso}_l\bW_l   - \by\by^\top\right)\right\}_{0\le k\le K} + \tr(\bW_k \mE)_{0\le k\le K}.
	\end{align*}
	Since $ p^{-1} \max_{0\le k\le K} |\tr(\bW_k\mE)| \le \lambda_0 / 2$ by the assumption, we are left with bounding the first term.
	The optimality of $\wh\bbeta^{\lasso}$ implies that
	\begin{equation*}
		\frac{1}{2p}\left\|\by\by^\top - \sum_{l=0}^K \wh\beta^{\lasso}_l\bW_l \right\|_F^2 +\lambda_0  \| \wh \bbeta^{\lasso} \|_1 \le  \frac{1}{2p}\left\|\by\by^\top - \sum_{l=0}^K \wh\beta^{\lasso}_l\bW_l - t\bW_k \right\|_F^2 +\lambda_0  \| \wh \bbeta^{\lasso} \|_1 + \lambda_0 |t|,
	\end{equation*}
	for any $t \in \mR$ and $0\le k\le K$.
	Then we have
	\begin{align*}
		\frac{t}{p}\tr\left\{\bW_k \left(\by\by^\top - \sum_{l=0}^K \wh\beta^{\lasso}_l\bW_l\right)\right\} \le \frac{t^2}{2p} \|\bW_k\|_F^2 + \lambda_0 |t| \le \frac{w^2 t^2}{2} + \lambda_0 |t|,
	\end{align*}
	where we have used Condition (C\ref{cond:bounded_norm}) and $\|\bW_k\|_F^2 \le p \|\bW_k\|_1^2 \le pw^2$ in the last inequality.
	Since $t$ is arbitrary, we conclude that $\left|\tr\left\{\bW_k \left(\by\by^\top - \sum_{l=0}^K \wh\beta^{\lasso}_l\bW_l\right)\right\} \right| \le \lambda_0$ for each $0\le k \le K$.
	Arranging these results, we conclude that
	\begin{align*}
		\|\wh\bbeta^{\lasso}-\bbeta^{(0)}\|_\infty = \|\wh\bdelta\|_\infty \le \frac{1}{p\kappa'} \|\bSigma_{W} \wh\bdelta \|_{\infty}\le  \frac{1}{\kappa'}(\frac{\lambda_0}{2} + \lambda_0) = \frac{3}{2\kappa'}\lambda_0.
	\end{align*}
	This gives the desired $\ell_\infty$-bound for the Lasso estimator. We can see that the error bound $ \|\wh\bbeta^{\lasso}-\bbeta^{(0)}\|_\infty =O (\lambda_0)$ is free of the factor $\sqrt{s}$.

	\subsection{Proof of Theorem \ref{thm:convergence}}
	\label{proof of thm:convergence}
	
	Following the idea of \cite{fan2014strong}, we prove the results in two steps.
	In the first step, we prove that the LLA algorithm converges under the given event.
	In the second step, we give the upper bounds for the three probabilities.
	In the last step, we show that the LLA algorithm converges to the oracle estimator with probability tending to one under the assumed conditions.
	
	\noindent
	{\bf Step 1.}
	Recall that $a_0 = \min\{1, a_2 \}$. 
	We first define three events as
	\begin{align*}
		&E_0 = \Big\{\| \wh \bbeta^{\initial} - \bbeta^{(0)} \|_{\infty} \le a_0 \lambda \Big\},\\
		&E_1 = \Big\{\|\nabla_{\mS^c} Q(\wh \bbeta^{\oracle}_{\mS})\|_{\infty} < a_1 \lambda\Big\},\\
		&E_2 = \Big\{\|\wh \bbeta^{\oracle}_{\mS} \|_{\min} \ge \gamma\lambda\Big\}.
	\end{align*}
	In the following, we prove that the LLA algorithm converges under the event $E_1\cap E_2 \cap E_3$ in two further steps.
	We first show that the LLA algorithm initialized by $\wh \bbeta^{\initial}$ finds $\wh \bbeta^{\oracle}$ after one iteration, under the event $E_0 \cap E_1$.
	We next show that if $\wh \bbeta^{\oracle}$ is obtained, then the LLA algorithm will find $\wh \bbeta^{\oracle}$ again in the next iteration, under the event $E_1 \cap E_2$.
	Then, we can immediately obtain that the LLA algorithm initialized by $\wh \bbeta^{\initial}$ should converge to $\wh \bbeta^{\oracle}$ after two iterations with probability at least $ P(E_0 \cap E_1 \cap E_2) \ge 1 - P(E_0^c)-P(E_1^c)-P(E_2^c) = 1-\delta_0-\delta_1-\delta_2$.
	
	\noindent
	{\bf Step 1.1.}
	Recall that $\wh\bbeta^{(0)} = \wh \bbeta^{\initial}$. Under the event $E_0$, due to Assumption \ref{cond:minimal_signal}, we have $\wh\beta_k^{(0)}\le \| \wh \bbeta^{(0)} - \bbeta^{(0)} \|_{\infty} \le a_0 \lambda\le a_2\lambda$ for $k \in \mS^c$, and $\wh \beta_k^{(0)} \ge \|\bbeta_{\mS}^{(0)}\|_{\min} - \| \wh \bbeta^{(0)} - \bbeta^{(0)} \|_{\infty} > \gamma\lambda$ for $k \in \mS$.
	By property (iv) of $p_\lambda(\cdot)$, we have $p'_\lambda(|\wh \beta_k^{(0)}|) = 0 $ for $k \in \mS$.
	Thus, according to step (2.a) of the Algorithm \ref{alg:LLA}, $\wh\bbeta^{(1)}$ should be the solution to the problem
	\begin{align}\label{eq:first_iter}
		\wh\bbeta^{(1)} = \argmin_{\bbeta} Q(\bbeta) + \sum_{k\in \mS^c } p'_\lambda(|\wh\bbeta_k^{(0)}|) |\beta_k|.
	\end{align}
	By properties (ii) and (iii), $p'_\lambda(|\wh\bbeta_k^{(0)}|) \ge a_1 \lambda$ holds for $k\in\mS^c$.
	We next show that $\wh \bbeta^{\oracle}$ is the unique global solution to \eqref{eq:first_iter} under the event $E_1$.
	By Condition (C\ref{cond:minimal_eigen}), we can verify that $\wh\beta^{\oracle}$ is the unique solution to $\argmin_{\bbeta: \bbeta_{\mS^c=\zero}} Q(\bbeta)$ and
	\begin{equation}\label{eq:zero_gradient}
		\nabla_{\mS} Q(\wh \bbeta^{\oracle}) \defeq \Big(\nabla_{k} Q(\wh \bbeta^{\oracle}), k\in\mS \Big) = \zero.
	\end{equation}
	Thus, for any $\bbeta$ we have
	\begin{align}\label{eq:convexity_Q2}
		\begin{aligned}
			Q(\bbeta) \ge& Q(\wh \bbeta^{\oracle}) + \sum_{k=0}^K \nabla_k Q(\wh \bbeta^{\oracle}) (\beta_k - \wh \bbeta_{k}^{\oracle})\\
			=&Q(\wh \bbeta^{\oracle}) + \sum_{k\in \mS^c} \nabla_k Q(\wh \bbeta^{\oracle}) (\beta_k - \wh \bbeta_{k}^{\oracle}).
		\end{aligned}
	\end{align}
	By \eqref{eq:convexity_Q2}, $\wh \bbeta_{\mS^c}^{\oracle} = \zero$ and under the event $E_1$, for any $\bbeta$ we have
	\begin{align*}
		&\left\{Q(\bbeta) + \sum_{k\in \mS^c } p'_\lambda(|\wh\beta_k^{(0)}|) |\beta_k| \right\} - \left\{Q(\wh \bbeta^{\oracle}) + \sum_{k\in \mS^c } p'_\lambda(|\wh\beta_k^{(0)}|) |\wh\beta_k^{\oracle}| \right\} \\
		\ge& \sum_{k \in \mS^c} \left\{p'_\lambda(|\wh\beta_k^{(0)}|) +  \nabla_k Q(\wh \bbeta^{\oracle}) \sign(\beta_k)\right\} |\beta_k|\\
		\ge& \sum_{k \in \mS^c} \left\{a_1 \lambda + \nabla_k Q(\wh \bbeta^{\oracle}) \sign(\beta_k) \right\} |\beta_k| \ge 0.
	\end{align*}
	The strict inequality holds unless $\beta_k =0 $ for all $k \in \mS^c$. By uniqueness of the oracle estimator, we should have $\wh \bbeta^{\oracle}$ is the unique solution to \eqref{eq:first_iter}.
	This proves $\wh\bbeta^{(1)} = \wh \bbeta^{\oracle}$.
	
	\noindent
	{\bf Step 1.2.} Given the LLA algorithm finds the oracle estimator, we denote $\wh \bbeta $ as the solution to the optimization problem in the next iteration of the LLA algorithm.
	By using $\wh \bbeta_{\mS^c}^{\oracle} = \zero$ and  $\nabla_k Q(\wh \bbeta^{\oracle}) = 0, \forall k \in \mS$, then under the event $E_2= \big\{\|\wh \bbeta^{\oracle}_{\mS} \|_{\min} \ge \gamma\lambda\big\}$, we have
	\begin{equation}\label{eq:2nd_iter}
		\wh\bbeta = \argmin_{\bbeta} Q(\bbeta) + \sum_{k\in \mS^c } p'_\lambda(0) |\beta_k|.
	\end{equation}
	Recall that $p'_\lambda(0) \ge a_1 \lambda$.
	Then by similar procedures in Step 1, we can show that $\wh \bbeta^{\oracle}$ is the unique solution to \eqref{eq:2nd_iter}, under the event $E_1=\big\{\|\nabla_{\mS^c} Q(\wh \bbeta^{\oracle}_{\mS})\|_{\infty} < a_1 \lambda\big\}$.
	Hence, the LLA algorithm converges, under the event $E_1\cap E_2$.
	This completes the proof of Step 1.
	
	\noindent
	{\bf Step 2.}
	We next give the upper bounds for $\delta_0 = P(E_0^c)$, $\delta_1 = P(E_1^c)$ and $\delta_2 = P(E_2^c)$ under the additional conditions.
	The three bounds are derived in the three further steps.
	
	\noindent
	{\bf Step 2.1.} Note that we use $\wh\bbeta^{\lasso}$ as the initial estimator.
	Then by Theorem \ref{thm:lasso_est} and the condition $\lambda\ge (3\sqrt{s+1}\lambda_0 ) / (a_0 \kappa)$, we have
	\begin{align*}
		\| \wh\bbeta^{\initial} - \bbeta^{(0)}\|_{\infty} \le \| \wh\bbeta^{\lasso} - \bbeta^{(0)}\| \le \frac{3 }{\kappa}\sqrt{s+1} \lambda_0 \le a_0 \lambda
	\end{align*}
	holds with probability at least $1-\delta_0'$ with
	\begin{align*}
		\delta_0' = 2(K+1)\exp\left\{ -\min\left( \frac{C_1 p \lambda_0^2}{w^2\sigma_{\max}^2},\frac{C_2 p \lambda_0}{w\sigma_{\max}} \right)\right\}.
	\end{align*}
	Consequently, we should have $\delta_0 = P(E_0^c) = P(\|\wh\bbeta^{\initial} - \bbeta^{(0)}\|_{\infty} > a_0 \lambda)\le \delta_0'$.
	This completes the proof of Step 2.1.

	\noindent
	{\bf Step 2.2.}
	We next bound the probability $\delta_1 = P(E_1^c)=P\big(\|\nabla_{\mS^c} Q(\wh \bbeta^{\oracle}_{\mS})\|_{\infty}\ge a_1 \lambda\big)$.
	Let $\bY = \vec(\by\by^\top)\in \mR^{p^2}$, $\bE = \vec(\mE)\in \mR^{p^2}$, and $\bV_k = \vec(\bW_k)\in \mR^{p^2}$.
	Further define $\V = (\bV_k: 1\le k\le K)\in\mR^{p^2 \times K}$, $\V_{\mS} = (\bV_k: k\in\mS)\in\mR^{p^2 \times (s+1)}$, and $\V_{\mS^c} = (\bV_k: k\in\mS^c)\in\mR^{p^2 \times (K-s)}$.
	Then we should have $\bY = \V_{\mS}\bbeta_{\mS}^{(0)} +\bE$, and $Q(\bbeta) = (2 p)^{-1}\|\bY - \V\bbeta \|^2$.
	Let $\H_{\mS} \defeq \V_{\mS}(\V_{\mS}^\top \V_{\mS})^{-1} \V_{\mS}^\top \in \mR^{p^2 \times p^2}$.
	Then we can compute that $\nabla_{\mS^c}Q(\wh \bbeta^{\oracle} ) = \big\{\nabla_{k}Q(\wh \bbeta^{\oracle} ) , k\in \mS^c\big\} = - p^{-1}\V_{\mS^c}^\top (\bI_{p^2} - \H_{\mS}) \bE $.
	By union bound, we have
	\begin{align}
		\delta_1 =& P\big(\|\nabla_{\mS^c} Q(\wh \bbeta^{\oracle}_{\mS})\|_{\infty}\ge a_1 \lambda\big)
		\le \sum_{k \in \mS^c} P\Big(|\bV_k^\top (\bI_{p^2} - \H_{\mS}) \bE | \ge p a_1  \lambda\Big)\nonumber \\
		\le& \sum_{k\in \mS^c} \bigg\{ P\Big(|\bV_k^\top\bE | \ge p a_1  \lambda / 2\Big) + P\Big(|\bV_k^\top \H_{\mS} \bE | \ge p a_1  \lambda/2 \Big) \bigg\}. \label{eq:delta_1_1}
	\end{align}
	Note that $\bV_k^\top\bE  = \tr(\bW_k \mE) = \tr\{\bW_k (\by\by^\top - \bSigma_0)\} = \by^\top \bW_k \by - \tr(\bW_k \bSigma_0)$. Then by Lemma \ref{lemma:hanson_wright} and Conditions (C\ref{cond:distribution}) and (C\ref{cond:bounded_norm}), we have $P\Big(|\bV_k^\top\bE | \ge p a_1  \lambda / 2\Big)=$
	\begin{align*}
		P\Big( \big|\by^\top \bW_k \by - \tr(\bW_k \bSigma_0)\big| > pa_1\lambda/2 \Big) \le 2\exp\left\{-\min\left(\frac{C_3 a_1^2 p \lambda^2 }{w^2\sigma_{\max}^2 }, \frac{C_4 a_1 p \lambda}{ w\sigma_{\max}}\right) \right\}.
	\end{align*}
	By Condition (C\ref{cond:bounded_norm}) and inequality \eqref{ineq:matrix_21} in Lemma \ref{lemma:norm_ineq}, we have $\|\bW_k\|\le\|\bW_k\|_1\le w$ for each $1\le k \le K$.
	Then we can derive that
	\begin{align*}
		|\bV_k^\top \H_{\mS} \bE |
		\le & \| (\V_{\mS}^\top\V_{\mS})^{-1} \V_{\mS}^\top\bV_k \|\|\V_{\mS}^\top \bE\|
		\le \| (\V_{\mS}^\top\V_{\mS})^{-1}\| \|\V_{\mS}^\top\bV_k\| \| \V_{\mS}^\top \bE\|  \\
		\le& \| \bSigma_{W, \mS}^{-1}\| \Big\{\sqrt{s+1}\max_{l \in \mS}|\tr(\bW_l \bW_k)|\Big\} \Big\{\sqrt{s+1}\max_{l \in \mS}|\tr(\bW_l \mE)|\Big\} \\
		\le&\Big\{ (p\tau_{\min})^{-1} \Big\} \Big\{\sqrt{s+1} (pw^2)\Big\} \Big\{\sqrt{s+1} \max_{l \in \mS}|\tr(\bW_l \mE)|\Big\}\\
		=  & \tau_{\min}^{-1} w^2 (s+1) \max_{l \in \mS} \big|\by^\top \bW_l \by - \tr(\bW_l \bSigma_0)\big|,
	\end{align*}
	where the third inequality is due to inequality \eqref{ineq:vector_max2} in Lemma \ref{lemma:norm_ineq}, and the last inequality is due to the following two facts: (i) by Condition (C\ref{cond:bounded_norm}) and inequality \eqref{ineq:matrix_21} in Lemma \ref{lemma:norm_ineq}, we have $|\tr(\bW_l \bW_k)| \le p\|\bW_l\|\|\bW_k\|\le pw^2$; (ii) by Condition (C\ref{cond:minimal_eigen}), we have $\big\|\bSigma_{W, \mS}^{-1} \big\| = \lambda_{\min}^{-1}(\bSigma_{W, \mS})\le (p\tau_{\min})^{-1}$.
	Then by Lemma \ref{lemma:hanson_wright} and Conditions (C\ref{cond:distribution}) and (C\ref{cond:bounded_norm}), we have $P\Big(|\bV_k^\top\H_{\mS} \bE | \ge p^2 a_1  \lambda / 2\Big)\le$
	\begin{align*}
		&\sum_{l \in \mS} P\left\{ \big|\by^\top \bW_l \by - \tr(\bW_l \bSigma_0)\big| > \frac{a_1 \tau_{\min} p\lambda }{2(s+1) w^2} \right\}\\
		\le& 2(s+1) \exp\left[-\min\left\{\frac{C_5 a_1^2 \tau_{\min}^2 p \lambda^2}{w^6\sigma_{\max}^2(s+1)^2}, \frac{C_6 a_1 \tau_{\min}p \lambda}{w^3 \sigma_{\max}(s+1)} \right\} \right]
	\end{align*}
	Together with \eqref{eq:delta_1_1}, we have
	\begin{align*}
		\delta_1 \le& 2 (K-s) \exp\left\{-\min\left(\frac{C_3 a_1^2 p \lambda^2 }{w^2\sigma_{\max}^2 }, \frac{C_4 a_1 p \lambda}{ w\sigma_{\max}}\right) \right\} \\ &+2(K-s)(s+1)\exp\left[-\min\left\{\frac{C_5 a_1^2 \tau_{\min}^2 p \lambda^2}{w^6\sigma_{\max}^2(s+1)^2}, \frac{C_6 a_1 \tau_{\min}p \lambda}{w^3 \sigma_{\max}(s+1)} \right\} \right].
	\end{align*}

	\noindent
	{\bf Step 2.3.}
	We next bound $\delta_2=P(E_2^c) = P(\|\wh \bbeta^{\oracle}_{\mS} \|_{\min} < \gamma\lambda)$. Note that $\wh \bbeta_{\mS}^{\oracle}  = \bbeta_{\mS}^{(0)} +(\V_{\mS}^\top \V_{\mS})^{-1} \V_{\mS}^\top\bE $, and thus $\|\wh \bbeta_{\mS}^{\oracle}\|_{\min} \ge \| \bbeta_{\mS}^{(0)}\|_{\min} -\|(\V_{\mS}^\top \V_{\mS})^{-1} \V_{\mS}^\top\bE\|_{\infty} $.
	Then we have
	\begin{equation}\label{eq:delta_2_1}
		\delta_2 \le P\Big(\|(\V_{\mS}^\top \V_{\mS})^{-1} \V_{\mS}^\top\bE\|_{\infty} \ge \| \bbeta_{\mS}^{(0)}\|_{\min} - \gamma \lambda\Big).
	\end{equation}
	Note that
	\begin{align*}
		&\|(\V_{\mS}^\top \V_{\mS})^{-1} \V_{\mS}^\top\bE\|_{\infty}
		\le \|(\V_{\mS}^\top \V_{\mS})^{-1}\V_{\mS}^\top\bE\|\le\|(\V_{\mS}^\top \V_{\mS})^{-1}\| \|\V_{\mS}^\top\bE\|\\
		\le & (p\tau_{\min})^{-1} \sqrt{s+1}\|\V_{\mS}^\top\bE\|_{\infty} = \sqrt{s+1}(p\tau_{\min})^{-1} \max_{k\in \mS}|\by^\top \bW_k \by - \tr(\bW_k \bSigma_0)| ,
	\end{align*}
	where the first inequality is due to inequality \eqref{ineq:vector_max2} in Lemma \ref{lemma:norm_ineq}, and the third inequality is due to Condition (C\ref{cond:minimal_eigen}) and \eqref{ineq:vector_max2} in Lemma \ref{lemma:norm_ineq}.
	Together with \eqref{eq:delta_2_1} and using Lemma \ref{lemma:hanson_wright}, we have
	\begin{align*}
		\delta_{2}\le& \sum_{k\in \mS} P\left\{ \big|\by^\top \bW_k \by - \tr(\bW_k \bSigma_0) \big| \ge \frac{\tau_{\min}p}{(s+1)^{1/2}}(\| \bbeta_{\mS}^{(0)}\|_{\min} - \gamma \lambda) \right\}\\
		\le&2 (s+1)\exp \left[-\min\left\{\frac{C_5 \tau_{\min}^2 p (\| \bbeta_{\mS}^{(0)}\|_{\min} - \gamma \lambda)^2 }{w^2\sigma_{\max}^2 (s+1)}, \frac{C_6 \tau_{\min} p (\| \bbeta_{\mS}^{(0)}\|_{\min} - \gamma \lambda) }{w\sigma_{\max}(s+1)^{1/2}} \right\} \right].
	\end{align*}
	This competes the proof of Step 2.
	
	\noindent{\bf Step 3.}
	To obtain the desired result, it suffices to prove that $\delta_1$, $\delta_2$, and $\delta_0'$ tend to $0$ as $p\to \infty$ under the assumed conditions.
	By Condition (C\ref{cond:minimal_signal}), we know that $\| \bbeta_{\mS}^{(0)}\|_{\min} - \gamma \lambda > \lambda$.
	Then, by inspecting the forms of upper bounds of $\delta_0,\ \delta_1, \delta_2$, it remains to prove that
	\begin{align}\label{eq:probs}
		\min\left\{\frac{p\lambda^2}{s^2}, \frac{p\lambda}{s}, \frac{p\lambda^2}{s}, \frac{p\lambda}{\sqrt{s}}, p\lambda_0^2, p\lambda_0,\right\} \Big/ \log(K) \to 0
	\end{align}
	as $p\to\infty$.
	Further note $\lambda\ge (3\sqrt{s+1}\lambda_0 ) / (a_0 \kappa)$.
	Then we can easily verify that, \eqref{eq:probs} holds as long as $p\lambda_0^2/\{s\log(K)\}\to\infty$ as $p\to \infty$.
	This completes the proof of Step 3 and completes the proof of the theorem.

	\subsection{Proof of Theorem \ref{thm:asymptotic_normality}}
	\label{proof of thm:asymptotic_normality}

	Recall that the oracle estimator is computed with the knowledge of the true support set of $\bbeta^{(0)}$. That is, $    \wh\bbeta^{\oracle} = \argmin_{\bbeta: \bbeta_{\mS^c=\zero}} Q(\bbeta)$, where $Q(\bbeta)$ is defined in \eqref{eq:Qbeta}.
	Equivalently, we should have
	\begin{align*}
		\wh\bbeta^{\oracle}_{\mS} - \bbeta^{(0)}_{\mS} = \bSigma_{W,\mS}^{-1} \bSigma_{WY, \mS}- \bbeta^{(0)}_{\mS} =  \bSigma_{W,\mS}^{-1} S_p,
	\end{align*}
	where $\bSigma_{W,\mS} = \{\tr(\bW_k\bW_l): k,l\in \mS \}\in \mR^{(s+1)\times (s+1)}$, $\bSigma_{WY,\mS} = \{\by^\top\bW_k\by: k\in\mS \}^\top \in \mR^{s+1}$, and
	\begin{align*}
		S_p =  
		\begin{pmatrix}
			\vec^\top( \bW_0 ) \\
			\vdots\\
			\vec^\top( \bW_s  )
		\end{pmatrix}\vec(\by\by^\top - \bSigma_0) =
		\begin{pmatrix}
			\vec^\top(\bSigma_0^{1/2} \bW_0 \bSigma_0^{1/2} ) \\
			\vdots\\
			\vec^\top(\bSigma_0^{1/2} \bW_s \bSigma_0^{1/2} )
		\end{pmatrix}\vec(\bZ\bZ^\top - \bI_p).
	\end{align*}
	Here we have used the facts that $\by = \bSigma^{1/2}\bZ$, and $\vec(\bM_1 \bM_2 \bM_3) = (\bM_3^\top \otimes \bM_1) \vec(\bM_2) $ for three arbitrary matrices $\bM_1$, $\bM_2$, $\bM_3$ of shapes $p_1\times p_2$, $p_2\times p_3$, and $p_3\times p_4$ \citep[see, e.g., (1.3.6) in][p. 28]{golub2013matrix}.
	Re-express $\bA = (\ba_1, \dots, \ba_L)^\top$, where $\ba_l = (a_{l0},\dots, a_{ls})^\top \in\mR^{s+1}$.
	Let $\wt S_p = (s+1)^{-1/2} \bA \bSigma_{W,\mS} ( \wh\bbeta^{\oracle}_{\mS} - \bbeta^{(0)}_{\mS}) = (s+1)^{-1/2}\bA S_p$.
	Then we should have
	\begin{align*}
		\wt S_p =
		\begin{pmatrix}
			\vec^\top(\bDelta_1 ) \\
			\vdots\\
			\vec^\top(\bDelta_L)
		\end{pmatrix}\vec(\bZ\bZ^\top - \bI_p) \in\mR^{L},
	\end{align*}
	where $\bDelta_l = (s+1)^{-1/2} \sum_{k=0}^s a_{lk} (\bSigma_0^{1/2} \bW_k \bSigma_0^{1/2} )$ for $1\le l\le L$.
	Further note that
	\begin{align*}
		\frac{1}{\sqrt{s+1}}\max_{1\le l\le L} \sum_{k=0}^s |a_{lk}| = \frac{1}{\sqrt{s+1}} \|\bA\|_{\infty} \le \|\bA\| <\infty,
	\end{align*}
	where the first inequality follows from \eqref{ineq:matrix_21} in Lemma \ref{lemma:norm_ineq}.
	By Condition (C\ref{cond:bounded_norm}), we have $\sup_{p,k} \|\bSigma_0^{1/2} \bW_k \bSigma_0^{1/2}\|_1 <\infty $.
	Then it follows that
	\begin{align*}
		\sup_{p} \| \bDelta_l\|_1 \le &
		\sup_{p} \frac{1}{\sqrt{s+1}}\sum_{k=0}^s |a_{lk}| \cdot \|\bSigma_0^{1/2} \bW_k \bSigma_0^{1/2}\|_1 \\
		\le&  \bigg\{\frac{1}{\sqrt{s+1}}\max_{1\le l\le L} \sum_{k=0}^s |a_{lk}|\bigg\} \bigg\{ \sup_{p,k}\|\bSigma_0^{1/2} \bW_k \bSigma_0^{1/2}\|_1\bigg\} < \infty,
	\end{align*}
	for each $1\le l\le L$.
	By using Lemma \ref{lemma:normality}, we know that
	\begin{align*}
		\cov(\wt S_p) = 2 \{\tr(\bDelta_k\bDelta_l):1\le l \le L \} + ( \mu_4 - 3) \{\tr(\bDelta_k \circ \bDelta_l): 1\le k,l\le L \}.
	\end{align*}
	By assumed conditions in the theorem, we can verify that $p^{-1}\cov(\wt S_p) \to \bC$. Then by Lemma \ref{lemma:normality}, we should have
	\begin{align*}
		\sqrt{p/(s+1)} \bA (p^{-1}\bSigma_{W,\mS}) ( \wh\bbeta^{\oracle}_{\mS} - \bbeta^{(0)}_{\mS}) = p^{-1/2}\wt S_p \to_d \mN(\zero, \bC).
	\end{align*}
	By Condition (C\ref{cond:convergence}), we know that $p^{-1} \bSigma_{W,\mS} \to \bG_0$ in the Frobenius norm.
	With the help of Slutsky's theorem, we obtain that $    \sqrt{p/(s+1)}\bA\bG_0 \Big(\wh\bbeta^{\oracle}_{\mS} - \bbeta_{\mS}^{(0)}\Big) \to_d \mN(\zero,\bC)$ as $p\to\infty$.
	This completes the proof of the theorem.

	\subsection{Proofs of Theorems \ref{thm:lasso_est_n} and \ref{thm:convergence_n}}
	\label{proof of additional thms}

	\textbf{Proof of Theorem \ref{thm:lasso_est_n}.}
	The proof is very similar to the proof of Theorem \ref{thm:lasso_est} in Appendix \ref{proof of thm:lasso_est}.
	Note that $\by_i\by_i^\top = \sum_{k=0}^K \beta_k^{(0)} \bW_k + \mE_i$ for $1\le i \le n$.
	Define $\wh \bdelta \defeq \wh\bbeta_n^{\lasso} - \bbeta^{(0)}$.
	We first show that, if $\lambda_0 \ge (2/p) \max_{0\le k\le K} |n^{-1}\sum_{i=1}^n\tr(\bW_k\mE_i)|$ holds, then $\wh\bdelta \in \C_3(\mS) \defeq \{\bdelta \in \mR^{K+1}: \|\bdelta_{\mS^c}\|_1 \le 3 \|\bdelta_{\mS}\|_1\}$.
	Subsequently, we show that $\big\{\lambda_0 \ge (2/p) \max_{0\le k\le K} |n^{-1}\sum_{i=1}^n\tr(\bW_k\mE_i)|\big\}$ holds with high probability.
	
	\noindent
	{\bf Step 1.}
	Since $\wh \bbeta_n^{\lasso}$ is the solution to $\argmin_{\bbeta} Q_n(\bbeta) + \lambda_0  \|\bbeta\|_1$, we have
	\begin{align*}
		Q_n(\wh\bbeta_n^{\lasso}) +\lambda_0  \| \wh \bbeta_n^{\lasso} \|_1
		=&  \frac{1}{2n p}\sum_{i=1}^n\left\|\mE_i - \sum_{k = 0}^K \wh\delta_k \bW_k \right\|_F^2 + \lambda_0  \| \wh \bbeta_n^{\lasso} \|_1 \\
		\le& \frac{1}{2n p}\sum_{i=1}^n\|\mE_i \|_F^2 + \lambda_0  \|\bbeta^{(0)} \|_1.
	\end{align*}
	Rearranging the above inequality, we obtain that
	\begin{equation}\label{ineq:lagrangian_n}
		0 \le \frac{1}{2p} \left\| \sum_{k=0}^K \wh\delta_k \bW_k \right\|_F^2 \le \frac{1}{np}\sum_{i=1}^n \tr\left( \mE_i \sum_{k=0}^K \wh\delta_k \bW_k \right) + \lambda_0 \Big\{\|\bbeta^{(0)}\|_1 - \|\wh\bbeta^{\lasso}\|_1 \Big\}
	\end{equation}
	Note that
	\begin{equation}\label{ineq:holder_n}
		\frac{1}{n}\sum_{i=1}^n\tr\left( \mE_i \sum_{k=0}^K \wh\delta_k \bW_k \right) \le \sum_{k=0}^K |\wh\delta_k | \cdot \Big|n^{-1}\sum_{i=1}^n\tr\left( \bW_k\mE_i \right) \Big| \le \|\wh\bdelta\|_1 \max_{0\le k\le K} \Big|n^{-1}\sum_{i=1}^n\tr\left( \bW_k\mE_i \right)\Big|.
	\end{equation}
	Since $\bbeta^{(0)}$ is supported on $\mS$, we can write $    \|\bbeta^{(0)}\|_1 - \|\wh\bbeta^{\lasso}\|_1 = \|\bbeta_{\mS}^{(0)}\|_1 - \|\bbeta_{\mS}^{(0)} + \wh\bdelta_{\mS}\|_1 - \|\wh\bdelta_{\mS^c}\|_1$.
	Substituting it into the inequality \eqref{ineq:lagrangian_n} and using the inequality \eqref{ineq:holder_n} yields
	\begin{align}
		0 \le& \frac{1}{p} \left\| \sum_{k=0}^K \wh\delta_k \bW_k \right\|_F^2 \le \frac{2}{p} \max_{0\le k\le K} \Big|n^{-1}\sum_{i=1}^n\tr(\bW_k\mE_i) \Big| \cdot \|\wh\bdelta\|_1+ 2\lambda_0 \Big\{\|\bbeta_{\mS}^{(0)}\|_1 - \|\bbeta_{\mS}^{(0)} + \wh\bdelta_{\mS}\|_1 - \|\wh\bdelta_{\mS^c}\|_1\Big\} \nonumber\\
		\le & \lambda_0 \|\wh\bdelta\|_1+ 2\lambda_0 \Big\{ \|\wh\bdelta_{\mS}\|_1 - \|\wh\bdelta_{\mS^c}\|_1\Big\} \le \lambda_0 \Big\{3 \|\wh\bdelta_{\mS}\|_1 - \|\wh\bdelta_{\mS^c}\|_1 \Big\}, \label{ineq:bound_n}
	\end{align}
	where we have used the condition $\lambda_0 \ge (2/p) \max_{0\le k\le K} |n^{-1}\sum_{i=1}^n\tr(\bW_k\mE_i)|$ in the third inequality.
	Thus, we conclude that $\wh\bdelta \in \C_3(\mS)$.
	Then, by the RE Condition (C\ref{cond:RE}) and the inequality \eqref{ineq:bound_n}, we can obtain that
	\begin{align*}
		\kappa \|\wh\bdelta\|^2 \le \frac{1}{p} \left\| \sum_{k=0}^K \wh\delta_k \bW_k \right\|_F^2 \le \lambda_0 \Big\{3 \|\wh\bdelta_{\mS}\|_1 - \|\wh \bdelta_{\mS^c}\|_1 \Big\} \le 3 \lambda_0 \sqrt{s+1} \| \wh\bdelta\|,
	\end{align*}
	where the last inequality follows from \eqref{ineq:vector_21} in Lemma \ref{lemma:norm_ineq} with $\|\wh\bdelta_{\mS}\|_1 \le \sqrt{s+1} \|\wh\bdelta_{\mS}\| \le \sqrt{s+1} \|\wh\bdelta\|$.
	This implies the conclusion $\|\wh\bbeta_n^{\lasso} - \bbeta^{(0)} \| = \|\wh\bdelta\| \le (3/\kappa)\sqrt{s+1} \lambda_0 $.
	
	\noindent
	{\bf Step 2.}
	It remains to show that the event $\big\{\lambda_0 \ge (2/p) \max_{0\le k\le K} |n^{-1}\sum_{i=1}^n\tr(\bW_k\mE_i)| \big\}$ holds with high probability.
	Recall that $n^{-1}\sum_{i=1}^n\tr(\bW_k\mE_i)=n^{-1}\sum_{i=1}^n\by_i^\top\bW_k\by_i - \tr(\bW_k\bSigma_0) $.
	Further note that Condition (C\ref{cond:bounded_norm}) and norm inequality \eqref{ineq:matrix_21} in Lemma \ref{lemma:norm_ineq} imply that $\sup_{p,k}\|\bW_k\| \le \sup_{p,k}\|\bW_k\|_1 \le w $ and $\|\bSigma_0\| \le \|\bSigma_0^{1/2}\|^2 \le \|\bSigma_0^{1/2}\|_1^2 \le \sigma_{\max}$.
	Then by union bound and Lemma \ref{lemma:hanson_wright}, we have
	\begin{align*}
		P\left\{ \frac{2}{p} \max_{0\le k\le K} |n^{-1}\sum_{i=1}^n\tr(\bW_k\mE_i)| \ge \lambda_0 \right\} \le& \sum_{k=0}^K P\left( \Big|n^{-1}\sum_{i=1}^n\by_i^\top\bW_k\by_i - \tr(\bW_k\bSigma_0) \Big| \ge \frac{p\lambda_0}{2}\right) \\
		\le& 2(K+1)\exp\left\{ -\min\left( \frac{C_1 n p \lambda_0^2}{w^2\sigma_{\max}^2},\frac{C_2 n p \lambda_0}{w\sigma_{\max}} \right)\right\}.
	\end{align*}
	Thus, we should have the event $\big\{\lambda_0 \ge (2/p) \max_{0\le k\le K} |n^{-1}\sum_{i=1}^n\tr(\bW_k\mE_i)| \big\}$ holds with the probability at least $1 - 2(K+1)\exp\left\{ -\min\left( \frac{C_1 n p \lambda_0^2}{w^2\sigma_{\max}^2},\frac{C_2 n p \lambda_0}{w\sigma_{\max}} \right)\right\}$. This completes the proof of the theorem.

	\noindent \textbf{Proof of Theorem \ref{thm:convergence_n}.}
	The proof is very similar to the proof of Theorem \ref{thm:convergence} in Appendix \ref{proof of thm:convergence}.
	There are three steps.
	In the first step, we need to prove that the LLA algorithm converges under the event $E_1\cap E_2 \cap E_3$, where
	\begin{align*}
		&E_0 = \Big\{\| \wh \bbeta_n^{\lasso} - \bbeta^{(0)} \|_{\infty} \le a_0 \lambda \Big\},\\
		&E_1 = \Big\{\|\nabla_{\mS^c} Q(\wh \bbeta^{\oracle}_{\mS})\|_{\infty} < a_1 \lambda\Big\},\\
		&E_2 = \Big\{\|\wh \bbeta^{\oracle}_{\mS} \|_{\min} \ge \gamma\lambda\Big\}.
	\end{align*}
	In the second step, we derive the upper bounds for $P(E_0^c)$, $P(E_1^c)$ and $P(E_2^c)$.
	In the last step, we show that the LLA algorithm converges to the oracle estimator with probability tending to one under the assumed conditions.
	Since the first step is almost the same as that in Appendix \ref{proof of thm:convergence}, we omit the details.

	\noindent
	{\bf Step 2.}
	In this step, we give the upper bounds for $\delta_0 = P(E_0^c)$, $\delta_1 = P(E_1^c)$ and $\delta_2 = P(E_2^c)$ under the assumed conditions.
	The three bounds are derived in the three further steps.
	
	\noindent
	{\bf Step 2.1.} Note that we use $\wh\bbeta_n^{\lasso}$ as the initial estimator.
	Then by Theorem \ref{thm:lasso_est_n} and the condition $\lambda\ge (3\sqrt{s+1}\lambda_0 ) / (a_0 \kappa)$, we have
	\begin{align*}
		\| \wh\bbeta_n^{\lasso} - \bbeta^{(0)}\|_{\infty} \le \| \wh\bbeta_n^{\lasso} - \bbeta^{(0)}\| \le \frac{3 }{\kappa}\sqrt{s+1} \lambda_0 \le a_0 \lambda
	\end{align*}
	holds with probability at least $1-\delta_0'$ with
	\begin{align*}
		\delta_0' = 2(K+1)\exp\left\{ -\min\left( \frac{C_1 n p \lambda_0^2}{w^2\sigma_{\max}^2},\frac{C_2 n p \lambda_0}{w\sigma_{\max}} \right)\right\}.
	\end{align*}
	Consequently, we should have $\delta_0 = P(E_0^c) = P(\|\wh\bbeta_n^{\lasso} - \bbeta^{(0)}\|_{\infty} > a_0 \lambda)\le \delta_0'$.
	This completes the proof of Step 2.1.

	\noindent
	{\bf Step 2.2.}
	We next bound the probability $\delta_1 = P(E_1^c)=P\big(\|\nabla_{n,\mS^c} Q(\wh \bbeta^{\oracle}_{\mS})\|_{\infty}\ge a_1 \lambda\big)$.
	Let $\bY_i = \vec(\by_i\by_i^\top)\in \mR^{p^2}$, $\bE_i = \vec(\mE_i)\in \mR^{p^2}$, and $\bV_k = \vec(\bW_k)\in \mR^{p^2}$.
	Further define $\V = (\bV_k: 1\le k\le K)\in\mR^{p^2 \times K}$, $\V_{\mS} = (\bV_k: k\in\mS)\in\mR^{p^2 \times (s+1)}$, and $\V_{\mS^c} = (\bV_k: k\in\mS^c)\in\mR^{p^2 \times (K-s)}$.
	Then we should have $\bY_i = \V_{\mS}\bbeta_{\mS}^{(0)} +\bE_i$, and $Q_n(\bbeta) = (2 np)^{-1}\sum_{i=1}^n\|\bY_i - \V\bbeta \|^2$.
	Let $\H_{\mS} \defeq \V_{\mS}(\V_{\mS}^\top \V_{\mS})^{-1} \V_{\mS}^\top \in \mR^{p^2 \times p^2}$, and $\wb{\bE} = n^{-1}\sum_{i=1}^n \bE_i$.
	Then we can compute that $\nabla_{\mS^c}Q(\wh \bbeta_n^{\oracle} ) = \big\{\nabla_{k}Q(\wh \bbeta_n^{\oracle} ) , k\in \mS^c\big\} = - p^{-1}\V_{\mS^c}^\top (\bI_{p^2} - \H_{\mS}) \wb{\bE}$.
	By union bound, we have
	\begin{align}
		\delta_1 =& P\big(\|\nabla_{\mS^c} Q(\wh \bbeta^{\oracle}_{\mS})\|_{\infty}\ge a_1 \lambda\big)
		\le \sum_{k \in \mS^c} P\Big(|\bV_k^\top (\bI_{p^2} - \H_{\mS}) \wb\bE | \ge p a_1  \lambda\Big)\nonumber \\
		\le& \sum_{k\in \mS^c} \bigg\{ P\Big(|\bV_k^\top \wb\bE | \ge p a_1  \lambda / 2\Big) + P\Big(|\bV_k^\top \H_{\mS} \wb\bE | \ge p a_1  \lambda/2 \Big) \bigg\}. \label{eq:delta_1_1_n}
	\end{align}
	Note that $\bV_k^\top \wb\bE  = \tr(n^{-1}\sum_{i=1}^n \bW_k \mE_i) = \tr\{n^{-1}\sum_{i=1}^n\bW_k (\by_i\by_i^\top - \bSigma_0)\} = n^{-1}\sum_{i=1}^n\by_i^\top \bW_k \by_i - \tr(\bW_k \bSigma_0)$. Then by Lemma \ref{lemma:hanson_wright} and Conditions (C\ref{cond:distribution}) and (C\ref{cond:bounded_norm}), we have $P\Big(|\bV_k^\top \wb\bE | \ge p a_1  \lambda / 2\Big)=$
	\begin{align*}
		P\Big(n^{-1}\sum_{i=1}^n \big|\by_i^\top \bW_k \by_i - \tr(\bW_k \bSigma_0)\big| > pa_1\lambda/2 \Big) \le 2\exp\left\{-\min\left(\frac{C_3 a_1^2 n p \lambda^2 }{w^2\sigma_{\max}^2 }, \frac{C_4 a_1 n p \lambda}{ w\sigma_{\max}}\right) \right\}.
	\end{align*}
	By Condition (C\ref{cond:bounded_norm}) and inequality \eqref{ineq:matrix_21} in Lemma \ref{lemma:norm_ineq}, we have $\|\bW_k\|\le\|\bW_k\|_1\le w$ for each $1\le k \le K$.
	Then we can derive that
	\begin{align*}
		|\bV_k^\top \H_{\mS} \wb\bE |
		\le & \| (\V_{\mS}^\top\V_{\mS})^{-1} \V_{\mS}^\top\bV_k \|\|\V_{\mS}^\top \wb\bE\|
		\le \| (\V_{\mS}^\top\V_{\mS})^{-1}\| \|\V_{\mS}^\top\bV_k\| \| \V_{\mS}^\top \wb\bE\|  \\
		\le& \| \bSigma_{W, \mS}^{-1}\| \Big\{\sqrt{s+1}\max_{l \in \mS}|\tr(\bW_l \bW_k)|\Big\} \Big\{\sqrt{s+1}\max_{l \in \mS}|\tr(n^{-1}\sum_{i=1}^n \bW_l \mE_i)|\Big\} \\
		\le&\Big\{ (p\tau_{\min})^{-1} \Big\} \Big\{\sqrt{s+1} (pw^2)\Big\} \Big\{\sqrt{s+1} \max_{l \in \mS}|\tr(n^{-1}\sum_{i=1}^n\bW_l \mE_i)|\Big\}\\
		=  & \tau_{\min}^{-1} w^2 (s+1) \max_{l \in \mS} \big|n^{-1}\sum_{i=1}^n \by_i^\top \bW_l \by_i - \tr(\bW_l \bSigma_0)\big|,
	\end{align*}
	where the third inequality is due to inequality \eqref{ineq:vector_max2} in Lemma \ref{lemma:norm_ineq}, and the last inequality is due to the following two facts: (i) by Condition (C\ref{cond:bounded_norm}) and inequality \eqref{ineq:matrix_21} in Lemma \ref{lemma:norm_ineq}, we have $|\tr(\bW_l \bW_k)| \le p\|\bW_l\|\|\bW_k\|\le pw^2$; (ii) by Condition (C\ref{cond:minimal_eigen}), we have $\big\|\bSigma_{W, \mS}^{-1} \big\| = \lambda_{\min}^{-1}(\bSigma_{W, \mS})\le (p\tau_{\min})^{-1}$.
	Then by Lemma \ref{lemma:hanson_wright} and Conditions (C\ref{cond:distribution}) and (C\ref{cond:bounded_norm}), we have $P\Big(|\bV_k^\top\H_{\mS} \wb\bE | \ge p^2 a_1  \lambda / 2\Big)\le$
	\begin{align*}
		&\sum_{l \in \mS} P\left\{ \big|n^{-1}\sum_{i=1}^n \by_i^\top \bW_l \by_i - \tr(\bW_l \bSigma_0)\big| > \frac{a_1 \tau_{\min} p\lambda }{2(s+1) w^2} \right\}\\
		\le& 2(s+1) \exp\left[-\min\left\{\frac{C_5 a_1^2 \tau_{\min}^2 n p \lambda^2}{w^6\sigma_{\max}^2(s+1)^2}, \frac{C_6 a_1 \tau_{\min} n p \lambda}{w^3 \sigma_{\max}(s+1)} \right\} \right]
	\end{align*}
	Together with \eqref{eq:delta_1_1_n}, we have
	\begin{align*}
		\delta_1 \le& 2 (K-s) \exp\left\{-\min\left(\frac{C_3 a_1^2 p \lambda^2 }{w^2\sigma_{\max}^2 }, \frac{C_4 a_1 p \lambda}{ w\sigma_{\max}}\right) \right\} \\ &+2(K-s)(s+1)\exp\left[-\min\left\{\frac{C_5 a_1^2 \tau_{\min}^2 p \lambda^2}{w^6\sigma_{\max}^2(s+1)^2}, \frac{C_6 a_1 \tau_{\min}p \lambda}{w^3 \sigma_{\max}(s+1)} \right\} \right].
	\end{align*}

	\noindent
	{\bf Step 2.3.}
	We next bound $\delta_2=P(E_2^c) = P(\|\wh \bbeta_{n,\mS}^{\oracle} \|_{\min} < \gamma\lambda)$. Note that $\wh \bbeta_{n,\mS}^{\oracle}  = \bbeta_{\mS}^{(0)} +(\V_{\mS}^\top \V_{\mS})^{-1} \V_{\mS}^\top \wb\bE $, and thus $\|\wh \bbeta_{n,\mS}^{\oracle}\|_{\min} \ge \| \bbeta_{\mS}^{(0)}\|_{\min} -\|(\V_{\mS}^\top \V_{\mS})^{-1} \V_{\mS}^\top \wb\bE\|_{\infty} $.
	Then we have
	\begin{equation}\label{eq:delta_2_1_n}
		\delta_2 \le P\Big(\|(\V_{\mS}^\top \V_{\mS})^{-1} \V_{\mS}^\top \wb\bE\|_{\infty} \ge \| \bbeta_{\mS}^{(0)}\|_{\min} - \gamma \lambda\Big).
	\end{equation}
	Note that
	\begin{align*}
		&\|(\V_{\mS}^\top \V_{\mS})^{-1} \V_{\mS}^\top \wb\bE\|_{\infty}
		\le \|(\V_{\mS}^\top \V_{\mS})^{-1}\V_{\mS}^\top \wb\bE\|\le\|(\V_{\mS}^\top \V_{\mS})^{-1}\| \|\V_{\mS}^\top \wb\bE\|\\
		\le & (p\tau_{\min})^{-1} \sqrt{s+1}\|\V_{\mS}^\top \wb\bE\|_{\infty} = \sqrt{s+1}(p\tau_{\min})^{-1} \max_{k\in \mS}|n^{-1}\sum_{i=1}^n \by_i^\top \bW_k \by_i - \tr(\bW_k \bSigma_0)| ,
	\end{align*}
	where the first inequality is due to inequality \eqref{ineq:vector_max2} in Lemma \ref{lemma:norm_ineq}, and the third inequality is due to Condition (C\ref{cond:minimal_eigen}) and \eqref{ineq:vector_max2} in Lemma \ref{lemma:norm_ineq}.
	Together with \eqref{eq:delta_2_1_n} and using Lemma \ref{lemma:hanson_wright}, we have
	\begin{align*}
		\delta_{2}\le& \sum_{k\in \mS} P\left\{ \big|n^{-1}\sum_{i=1}^n \by_i^\top \bW_k \by_i - \tr(\bW_k \bSigma_0) \big| \ge \frac{\tau_{\min}p}{(s+1)^{1/2}}(\| \bbeta_{\mS}^{(0)}\|_{\min} - \gamma \lambda) \right\}\\
		\le&2 (s+1)\exp \left[-\min\left\{\frac{C_5 \tau_{\min}^2 n p (\| \bbeta_{\mS}^{(0)}\|_{\min} - \gamma \lambda)^2 }{w^2\sigma_{\max}^2 (s+1)}, \frac{C_6 \tau_{\min} n p (\| \bbeta_{\mS}^{(0)}\|_{\min} - \gamma \lambda) }{w\sigma_{\max}(s+1)^{1/2}} \right\} \right].
	\end{align*}
	This competes the proof of Step 2.
	
	\noindent{\bf Step 3.}
	To obtain the desired result, it suffices to prove that $\delta_1$, $\delta_2$, and $\delta_0'$ tend to $0$ as $p\to \infty$ under the assumed conditions.
	By Condition (C\ref{cond:minimal_signal}), we know that $\| \bbeta_{\mS}^{(0)}\|_{\min} - \gamma \lambda > \lambda$.
	Then, by inspecting the forms of upper bounds of $\delta_0,\delta_1, \delta_2$, it remains to prove that
	\begin{align}\label{eq:probs_n}
		\min\left\{\frac{n p\lambda^2}{s^2}, \frac{n p\lambda}{s}, \frac{n p\lambda^2}{s}, \frac{n p\lambda}{\sqrt{s}}, n p\lambda_0^2,
		n p\lambda_0,\right\} \Big/ \log(K) \to 0
	\end{align}
	as $p\to\infty$.
	Further note $\lambda\ge (3\sqrt{s+1}\lambda_0 ) / (a_0 \kappa)$.
	Then we can easily verify that, \eqref{eq:probs_n} holds as long as $n p\lambda_0^2/\{s\log(K)\}\to\infty$ as $np\to \infty$.
	This completes the proof of Step 3 and completes the proof of the theorem.

	\subsection{Useful Lemmas}
	\label{append:lemmas}
	
	\begin{lemma}\label{lemma:norm_ineq}
		{\sc(Norm Inequalities)} Let $\bv \in \mR^{p}$ be an arbitrary vector, and   $\bDelta \in \mR^{p\times p}$ be an arbitrary symmetric matrix. Then we should have
		\begin{align}
			&\|\bv\| \le \|\bv\|_1 \le \sqrt{p} \|\bv\|, \label{ineq:vector_21} \\
			&\|\bv\|_{\infty}\le \|\bv\|\le \sqrt{p} \|\bv\|_{\infty}, \label{ineq:vector_max2}\\
			&\|\bDelta\| \le \|\bDelta\|_F \le \sqrt{p} \|\bDelta\|, \label{ineq:matrix_2F}\\
			&\| \bDelta\| \le \|\bDelta\|_1 = \|\bDelta\|_\infty \le \sqrt{p}\|\bDelta\|. \label{ineq:matrix_21}
		\end{align}
	\end{lemma}
	\begin{proof}
		The inequalities \eqref{ineq:vector_21}, \eqref{ineq:vector_max2}, and \eqref{ineq:matrix_2F} are directly from (2.2.5), (2.2.6), and (2.3.7) in \cite[][p. 69, 72]{golub2013matrix}, respectively.
		Since $\bDelta$ is symmetric, we immediately obtain that $\|\bDelta\|_1 = \|\bDelta\|_\infty$ by definitions of the two norms; see for example (2.3.9) and (2.3.10) in \cite[][p. 72]{golub2013matrix}.
		Then by Corollary 2.3.2  in \cite[][p. 73]{golub2013matrix}, we have
		\begin{align*}
			\|\bDelta\| \le \sqrt{\|\bDelta\|_1\|\bDelta\|_\infty} = \|\bDelta\|_1=\bDelta\|_\infty.
		\end{align*}
		The rightmost inequality $ \|\bDelta\|_\infty \le \sqrt{p}\|\bDelta\|$ follows from  (2.3.11) in \cite[][p. 72]{golub2013matrix}.
		This completes the proof.
	\end{proof}

	\begin{lemma} \label{lemma:hanson_wright}
		{\sc (Hanson-Wright Inequality)} Let $\by = \bSigma^{1/2}\bZ$, where $\bZ = (Z_1,\dots,Z_p)^\top \in \mR^p$ is a random vector with independent and identically distributed sub-Gaussian coordinates. Assume that $E(Z_j)=0$, $\var(Z_j)=1$ for each $1\le j\le p$, and $\bSigma\in\mR^{p\times p}$ is a positive definite matrix.
		Let $\bDelta \in\mR^{p\times p}$ be a symmetric matrix. Then, for every $t\ge 0$, we have
		\begin{align*}
			P\Big\{ \big|\by^\top \bDelta \by - \tr(\bDelta \bSigma) \big| \ge t   \Big\}  \le  2\exp\left\{-  \min\left(\frac{C_1t^2}{p\|\bDelta\|^2\|\bSigma\|^2}, \frac{C_2 t}{\|\bDelta\|\|\bSigma\|} \right) \right\},
		\end{align*}
		where $C_1$ and $C_2$ are two positive constants.
		{Furthermore, suppose that $\by_i\ (1\le i\le n)$ are $n$ independent copies of $\by$, then we have
			\begin{align*}
				P\Big\{ \Big|n^{-1}\sum_{i=1}^n\by_i^\top \bDelta \by_i - \tr(\bDelta \bSigma) \Big| \ge t   \Big\}  \le  2\exp\left\{-  \min\left(\frac{C_1nt^2}{p\|\bDelta\|^2\|\bSigma\|^2}, \frac{C_2 nt}{\|\bDelta\|\|\bSigma\|} \right) \right\}.
		\end{align*}}
		\begin{proof}
			By using ordinary Hanson-Wright inequality \citep[e.g., Theorem 6.2.1 in][]{vershynin2018high}, we have $P\big\{ \big|\by^\top \bDelta \by - \tr(\bDelta \bSigma) \big| \ge t  \big\} =$
			\begin{align*}
				P\Big\{ \big|\bZ^\top(\bSigma^{1/2} \bDelta\bSigma^{1/2}) \bZ - \tr(\bDelta \bSigma) \big| \ge t   \Big\} \le 2\exp\left\{-  \min\left(\frac{C_1t^2}{\|\bSigma^{1/2} \bDelta\bSigma^{1/2}\|_F^2}, \frac{C_2 t}{\|\bSigma^{1/2} \bDelta\bSigma^{1/2}\|} \right) \right\}.
			\end{align*}
			By norm inequality \eqref{ineq:matrix_2F} in Lemma \ref{lemma:norm_ineq}, we have $\|\bSigma^{1/2} \bDelta\bSigma^{1/2}\|_F^2 \le p\|\bSigma^{1/2} \bDelta\bSigma^{1/2}\|^2$.
			Further note that $\|\bSigma^{1/2} \bDelta\bSigma^{1/2}\|\le \|\bSigma^{1/2}\|^2 \| \bDelta\| = \|\bDelta\|\|\bSigma\|$.
			Then we can immediately obtain the first inequality of the lemma.
			
			We next prove the second inequality of the lemma.
			Note that $\by_i = \bSigma^{1/2}\bZ_i$, where $\bZ_i\ (1\le i\le n)$ are $n$ independent and identically distributed random vectors, and $\mZ = (\bZ_1^\top,\dots, \bZ_n^\top)^\top\in\mR^{np}$ independent and identically distributed sub-Gaussian coordinates.
			Denote $\mA = \bI_n \otimes (\bSigma^{1/2} \bDelta\bSigma^{1/2}) \in \mR^{(np)\times (np)}$.
			Then, by using ordinary Hanson-Wright inequality, we have
			\begin{align*}
				&P\big\{ \big| n^{-1}\sum_{i=1}^n \by_i^\top \bDelta \by_i - \tr(\bDelta \bSigma) \big| \ge t  \big\} =P\bigg\{ \Big| \sum_{i=1}^n \bZ_i^\top (\bSigma^{1/2} \bDelta\bSigma^{1/2})\bZ_i -n\tr(\bDelta \bSigma) \Big|>nt  \bigg\}\\
				=&  P\bigg\{ \Big| \mZ^\top \mA \mZ - \tr(\mA) \Big|>nt  \bigg\}
				\le 2\exp\left\{-  \min\left(\frac{C_1 n^2 t^2}{\|\mA\|_F^2}, \frac{C_2 nt}{\|\mA\|} \right) \right\}.
			\end{align*}
			By using the relationship between matrix norm and Kronecker product \citep[e.g., results on Page 709 of][]{golub2013matrix}, we have $\|\mA\|_F^2 = \|\bI_n\|_F^2 \|\bSigma^{1/2} \bDelta\bSigma^{1/2}\|_F^2 \le np\|\bDelta\|^2\|\bSigma\|^2$, and $\|\mA\| =\|\bI_n\| \|\bSigma^{1/2} \bDelta\bSigma^{1/2}\|\le \|\bDelta\|\|\bSigma\|$.
			Then we can immediately obtain the second inequality of the lemma. This completes the proof of the lemma.
			
		\end{proof}
		
	\end{lemma}
	
	\begin{lemma} \label{lemma:normality}
		Let $\bZ = (Z_1, \dots, Z_p)^\top \in \mR^{p}$, where $Z_1, \dots, Z_p$ are independent and identically distributed with mean $0$ and variance $1$. Define
		\begin{align*}
			S_p = \begin{pmatrix}
				\vec^\top(\bDelta_1) \\
				\vdots\\
				\vec^\top(\bDelta_L)
			\end{pmatrix}\vec(\bZ\bZ^\top - \bI_p),
		\end{align*}
		where $\bDelta_l\in\mR^{p\times p}$ is a symmetric matrix for $1\le l \le L$ with $L<\infty$. Suppose that $\sup_{p} \| \bDelta_l\|_1 < \infty$ for $1\le l\le L$, and $E|Z_j|^{4+\eta}<\infty$ for some $\eta>0$. Then we have
		$E(S_p)=0$, and
		\begin{align*}
			\cov(S_p) = 2 \{\tr(\bDelta_k\bDelta_l):1\le l \le L \} + ( \mu_4 - 3) \{\tr(\bDelta_k \circ \bDelta_l): 1\le k,l\le L \},
		\end{align*}
		where $\mu_4 = E(Z_j^4)$. Moreover, $p^{-1/2-\varepsilon}S_p\to^{L_2} 0$ for ant $\varepsilon>0$. In addition, assume that there is a positive definite matrix $\bV \in \mR^{L\times L}$ such that $p^{-1}\cov(S_p) \to \bV$, then we have $p^{-1/2} S_p \to_d \mN(\zero, \bV)$ as $p \to \infty$.
		
		\begin{proof}
			This is directly modified from Lemma 4 in the supplementary material of \cite{zou2021network}.
		\end{proof}

	\end{lemma}

	\subsection{Verification of Conditions (C\ref{cond:minimal_eigen}), (C\ref{cond:RE}), and (C\ref{cond:convergence})}
	\label{append:verify}
	
	We consider a specific example to verify Conditions (C\ref{cond:minimal_eigen}), (C\ref{cond:RE}), and (C\ref{cond:convergence}). Specifically, we assume that $\bW_k\ (1\le k\le K)$ are $K$ similarity matrices independently generated as follows.
	More specifically, assume that $\bW_k=(w_{k,j_1 j_2}) \in \mR^{p\times p}$ is a symmetric matrix, whose diagonal elements are set to be zeros, and off-diagonal elements are independently and identically generated from Bernoulli distributions with probability $\theta  / (p-1) \in (0,1)$ for some constant $\theta\ge 1$.
	We then have the following lemma, which is useful for the subsequent verification of the conditions.

	\begin{lemma}\label{lemma:omega_k1k2}
		Let $\wh\omega_{ k_1 k_2}=p^{-1}\tr(\bW_{k_1} \bW_{k_2})$ for each $1\le k_1, k_2\le K$. Then for any $t\ge 0$, we have
		\begin{align}
			P\Big(| \wh\omega_{kk}-\theta) | \ge t \Big) \le 2\exp\left\{-\frac{p t^2}{4\theta+ 4t/3} \right\} \label{ineq:omega_kk},
		\end{align}
		for any $1\le k\le K$.
		In addition, for any $t\ge 2\theta^2/p$, we have
		\begin{align}
			P\Big( | \wh\omega_{k_1 k_2}| \ge   t \Big) \le
			2\exp\left\{-\frac{p \big(t- 2\theta^2/p\big)^2}{4\theta^2 + 4t/3} \right\}, \label{ineq:omega_k1k2}  \end{align}
		for any $ k_1 \ne k_2$.
		
	\end{lemma}
	\begin{proof}
		We first prove \eqref{ineq:omega_kk}.
		In fact, we can compute that $\wh\omega_{kk}=p^{-1}\tr(\bW_k^2) = 2p^{-1} \sum_{j_1>j_2} w_{k,j_1 j_2}^2= 2p^{-1} \sum_{j_1>j_2} w_{k,j_1 j_2}$, since $w_{k,j_1 j_2}$s are Bernoulli random variables.
		Note that $E (w_{k,j_1 j_2}) = \theta / (p-1)$ and $\var(w_{k,j_1 j_2}) = \{\theta / (p-1)\} \{1- \theta / (p-1)\} \le \theta / (p-1)$.
		Then by Bernstein's inequality for sum of independent bounded random variables  \citep[e.g., Theorem 2.8.4 in][]{vershynin2018high}, we have
		\begin{align*}
			P\left(\bigg| \sum_{j_1>j_2}  \Big(w_{k,j_1 j_2} - \frac{\theta}{p-1}\Big) \bigg| \ge t \right) \le 2\exp\left\{-\frac{ t^2 / 2}{p\theta /2 + t/3} \right\},
		\end{align*}
		for any $t\ge 0$. By Replacing $t$ with $pt / 2$, we can directly obtain \eqref{ineq:omega_kk}.
		
		We next prove \eqref{ineq:omega_k1k2}.
		Note that $\wh\omega_{k_1 k_2}=p^{-1} \tr(\bW_{k_1}\bW_{k_2}) = 2p^{-1} \sum_{j_1>j_2} w_{k_1,j_1 j_2} w_{k_2,j_1 j_2}$.
		Then it is easy to compute that $E( w_{k_1,j_1 j_2} w_{k_2,j_1 j_2}) = \theta^2 / (p-1)^2$ and $\var( w_{k_1,j_1 j_2} w_{k_2,j_1 j_2})  \le \theta^2 / (p-1)^2$. Similarly, by using Bernstein's inequality we have
		\begin{align*}
			P\left(\bigg| \sum_{j_1>j_2}  \Big(w_{k_1,j_1 j_2} w_{k_2,j_1 j_2} - \frac{\theta^2}{(p-1)^2}\Big) \bigg| \ge t \right) \le 2\exp\left\{-\frac{ t^2 / 2}{\theta^2+ t/3} \right\},
		\end{align*}
		for any $t\ge 0$.  By Replacing $t$ with $pt / 2$, we can obtain that
		\begin{align*}
			P\Big( \Big| \wh\omega_{k_1 k_2}- \theta^2 / (p-1)  \Big| \ge t \Big) \le 2\exp\left\{-\frac{p t^2}{8\theta^2 /p+ 4t/3} \right\}.
		\end{align*}
		Then by using $(p-1)^{-1} \le 2/p$ for $p\ge 2$, we can derive that for any $t \ge 2\theta^2/p$,
		\begin{align*}
			P\Big( | \wh\omega_{k_1 k_2}| \ge   t \Big) \le P\Big( | \wh\omega_{k_1 k_2} - \theta^2 / (p-1) | \ge   t - \theta^2 / (p-1) \Big)  \le
			2\exp\left\{-\frac{p \big(t- 2\theta^2/p\big)^2}{4\theta^2 + 4t/3} \right\}.
		\end{align*}
		This proves \eqref{ineq:omega_k1k2} and completes the proof of the lemma.
	\end{proof}

	\noindent{\bf Verification of Condition (C\ref{cond:minimal_eigen})}. Define $\wh \bOmega_{\mS} =  p^{-1}\bSigma_{W, \mS} = (\wh\omega_{ k_1 k_2} ) \in \mR^{(s+1)\times (s+1)}$ with $\wh\omega_{k_1 k_2} = p^{-1}\tr(\bW_{k_1} \bW_{k_2})$ for $k_1, k_2 \in \mS$. Recall that $\bW_0 = \bI_p$. Then one can easily verify that $\wh\omega_{k0}=\wh\omega_{0k} = 1$ if $k=1$ and $\wh\omega_{k0}=\wh\omega_{0k} =0$ otherwise. Further define  $\bOmega_{\mS} = \diag\{1, \theta,\dots,\theta\} \in \mR^{(s+1)\times (s+1)}$. Then by Lemma \ref{lemma:omega_k1k2}, we know that
	\begin{align*}
		P\left\{ \| \wh\bOmega_{\mS} -\bOmega_{\mS} \|_{\max} \ge t  \right\} \le 2 s^2 \exp\left\{-\frac{p \big(t- 2\theta^2/p\big)^2}{4\theta^2  +  4t/3}  \right\},
	\end{align*}
	for any $t \ge 2\theta^2/p$. Here, $\| \bM\|_{\max} = \max_{i,j}|m_{ij}|$ denotes the element-wise max-norm for an arbitrary matrix $\bM = (m_{ij})$.
	This implies that $\bOmega_{\mS} $ should be the probabilistic limit of $\wh\bOmega_{\mS}$.
	By matrix norm inequality, we know that $\| \wh\bOmega_{\mS} -\bOmega_{\mS} \|  \le (s+1) \| \wh\bOmega_{\mS} -\bOmega_{\mS} \|_{\max} $.
	Since $2s\ge s+1$, we can deduce that
	\begin{align*}
		P\left\{ \| \wh\bOmega_{\mS} -\bOmega_{\mS} \| \ge t  \right\} \le P\left\{ \| \wh\bOmega_{\mS} -\bOmega_{\mS} \|_{\max} \ge t/(s+1)  \right\} \le 2 s^2 \exp\left\{-\frac{p \big\{t/(2s)- 2\theta^2/p\big\}^2}{4\theta^2  +  4t/3}  \right\},
	\end{align*}
	for any $t \ge 4\theta^2 s / p$.
	This implies that $\lambda_{\min}(\wh\bOmega_{\mS}) \ge \lambda_{\min}(\bOmega_{\mS}) - \| \wh\bOmega_{\mS} -\bOmega_{\mS} \| \to_p 1 $ as $p\to\infty$, provided $p / \{s^2\log(s)\} \to \infty$ as $p\to\infty$.
	Consequently, we should expect that Condition (C\ref{cond:minimal_eigen}) holds with high probability.

	\noindent{\bf Verification of Condition (C\ref{cond:RE})}.
	Similarly, define $\wh \bOmega =  p^{-1}\bSigma_{W} = (\wh\omega_{ k_1 k_2} ) \in \mR^{(K+1)\times (K+1)}$ with $\wh\omega_{k_1 k_2} = p^{-1}\tr(\bW_{k_1} \bW_{k_2})$ for $0\le k_1, k_2\le K$.
	Recall that $\bdelta \in \C_3(\mS) \defeq \{\bdelta \in \mR^{K+1}: \|\bdelta_{\mS^c}\|_1 \le 3 \|\bdelta_{\mS}\|_1\}$.
	Let $\mT \subset \mS^c$ collect the indexes of the $s+1$ largest $|\delta_k|$ in $\mS^c$. Further define $\wb \mS = \mS \cup \mT$. Then we should have
	\begin{align*}
		\frac{1}{p} \left\|\sum_{k=0}^K \delta_k \bW_k \right\|_F^2  =& \frac{1}{p} \left\|\sum_{k\in
			\wb\mS} \delta_k \bW_k \right\|_F^2 + 2\sum_{k_1 \in \wb\mS}\sum_{k_2\in \wb\mS^c} \delta_{k_1}\delta_{k_2} \wh\omega_{k_1 k_2} + \frac{1}{p} \left\|\sum_{k\in \wb\mS^c} \delta_k \bW_k \right\|_F^2\\
		\ge & \frac{1}{p} \left\|\sum_{k\in
			\wb\mS} \delta_k \bW_k \right\|_F^2 + 2\sum_{k_1 \in \wb\mS}\sum_{k_2\in \wb\mS^c} \delta_{k_1}\delta_{k_2} \wh\omega_{k_1 k_2} = Q_1 + Q_2.
	\end{align*}
	We next investigate $Q_1$ and $Q_2$, respectively.
	
	Let $\wh\bOmega_{\wb\mS} = \big(\wh\omega_{ k_1 k_2}: k_1,k_2\in \wb\mS\big) \in \mR^{(2s+2)\times (2s+2)}$ be the sub-matrix of $\wh\bOmega$. Similarly, let $\bOmega_{\wb\mS} = \diag\{1,\theta,\dots,\theta\}\in\mR^{(2s+2)\times (2s+2)}$. Then by similar procedures in the verification of Condition (C\ref{cond:minimal_eigen}), we can derive that  $\|\wh\bOmega_{\wb\mS}  - \bOmega_{\wb\mS} \| \to_p 0$ as long as $p / \{s^2\log(s)\} \to \infty$ as $p\to\infty$.
	Then it follows that
	\begin{align*}
		Q_1 = \frac{1}{p} \left\|\sum_{k\in \wb\mS} \delta_k \bW_k \right\|_F^2 = \bdelta_{\wb\mS}^\top \wh \bOmega_{\wb\mS} \bdelta_{\wb\mS} \ge \lambda_{\min}(\bOmega_{\wb\mS})\|\bdelta_{\wb\mS}\|^2 + \bdelta_{\wb\mS}^\top (\wh \bOmega_{\wb\mS} -\bOmega_{\wb\mS}  )\bdelta_{\wb\mS}  = \|\bdelta_{\wb\mS}\|^2 \{1+ o_p(1)\},
	\end{align*}
	as long as $p / \{s^2\log(s)\} \to \infty$ as $p\to\infty$.
	
	For the term $Q_2$, we can derive that
	\begin{align*}
		&|Q_2|=\left|2\sum_{k_1 \in \wb\mS}\sum_{k_2\in \wb\mS^c}
		\delta_{k_1}\delta_{k_2} \wh\omega_{k_1 k_2}\right|
		\le 4(s+1) \max_{k_1\in\wb\mS}|\delta_{k_1}| \cdot \max_{k_1\in\wb\mS, k_2 \in \wb\mS^c}|\wh\omega_{k_1 k_2}| \cdot \sum_{k_2\in \wb\mS^c} |\delta_{k_2}|\\
		\le & 4(s+1) \|\bdelta_{\wb\mS}\| \cdot \max_{k_1\in\wb\mS, k_2 \in \wb\mS^c}|\wh\omega_{k_1 k_2}| \cdot \|\bdelta_{\wb\mS^c}\|_1
		\le 12(s+1)^{3/2} \|\bdelta\|^2   \cdot \max_{k_1\in\wb\mS, k_2 \in \wb\mS^c}|\wh\omega_{k_1 k_2}|,
	\end{align*}
	where we have used the facts that $\|\bdelta_{\wb\mS}\|\le \|\bdelta\|$ and  $\|\bdelta_{\wb\mS^c}\|_1 \le \|\bdelta_{\mS^c}\|_1 \le 3 \|\bdelta_{\mS}\|_1 \le  3(s+1)^{1/2} \|\bdelta_{\mS}\| \le  3(s+1)^{1/2} \|\bdelta\|$.
	By \eqref{ineq:omega_k1k2} in Lemma \ref{lemma:omega_k1k2}, we know that
	\begin{align*}
		P\Big(\max_{k_1\in\wb\mS, k_2 \in \wb\mS^c}|\wh\omega_{k_1 k_2}|\ge t\Big) \le 4 (s+1) (K-2s-1)\exp\left\{-\frac{p \big(t- 2\theta^2/p\big)^2}{4\theta^2 + 4t/3} \right\},
	\end{align*}
	for any $t \ge 2\theta^2/p$.
	Hence, we should have $\max_{k_1\in\wb\mS, k_2 \in \wb\mS^c}|\wh\omega_{k_1 k_2}| = O_p(\sqrt{\log(Ks)/p})$.
	This indicates that $|Q_2| = o_p(\|\bdelta\|^2)$ as long as $ p / \{s^3 \log(Ks)\} \to \infty$ as $p\to \infty$.

	By far, we have shown that $p^{-1} \left\|\sum_{k=0}^K \delta_k \bW_k \right\|_F^2  \ge \|\bdelta_{\wb\mS}\|^2 \{1+ o_p(1)\} + o_p(\|\bdelta\|^2) = \|\bdelta_{\wb\mS}\|^2+ o_p(\|\bdelta\|^2)$.
	Thus, if we can show that $ \|\bdelta_{\wb\mS}\|^2\ge \kappa \|\bdelta\|^2$ for some $\kappa>0$ and $\bdelta \in \C_3(\mS)$, then Condition (C\ref{cond:RE}) should hold with high probability.
	In fact, by Lemma 2.2 of \cite{van2009conditions}, we have $\|\bdelta_{\wb\mS^c}\| \le (s+1)^{-1/2}\|\bdelta_{\mS^c}\|_1$. Since  $\bdelta \in \C_3(\mS)$, it follows that $ \|\bdelta_{\wb\mS^c}\|  \le 3(s+1)^{-1/2}\|\bdelta_{\mS}\|_1 \le 3 \|\bdelta_{\mS}\| \le  3 \|\bdelta_{\wb\mS}\|$, where we have used $\|\bdelta_{\mS}\|_1 \le (s+1)^{1/2}\|\bdelta_{\mS}\| $ in the second inequality.
	Then we should have $\|\bdelta\|^2 = \|\bdelta_{\wb\mS}\|^2 + \|\bdelta_{\wb\mS^c}\|^2 \le 10 \|\bdelta_{\wb\mS}\|^2$, or equivalently, $\|\bdelta_{\wb\mS}\|^2 \ge 0.1 \|\bdelta\|^2$.
	Combine above results, we can obtain that $p^{-1} \left\|\sum_{k=0}^K \delta_k \bW_k \right\|_F^2  \ge  0.1\|\bdelta\|^2+ o_p(\|\bdelta\|^2)$, as long as $ p / \{s^3 \log(Ks)\} \to \infty$ as $p\to \infty$. Thus, we should expect that RE Condition (C\ref{cond:RE}) holds with high probability.

	\noindent{\bf Verification of Condition (C\ref{cond:convergence})}.
	We consider a special case that $\bSigma_0 = \bSigma(\bbeta^{(0)}) = \beta_0^{(0)} \bI_p + \beta_1^{(0)}\bW_1 $ with $\beta_0^{(0)}, \beta_1^{(0)}>0$.
	By our above results, we can show that $\bG_{0,p} =  p^{-1}\bSigma_{W, \mS} \to_p \bG_0 \defeq \diag\{1, \theta \}$, which is positive definite.
	In addition, we have
	\begin{align*}
		\bG_{1,p} = p^{-1} \begin{bmatrix}
			\tr(\bSigma_0^2) & \tr(\bSigma_0^2\bW_1)\\
			\tr(\bSigma_0^2\bW_1) & \tr\{(\bSigma_0\bW_1)^2\}
		\end{bmatrix}.
	\end{align*}
	We next examine each entry of $\bG_{1,p}$.
	First, we can compute that $p^{-1}\tr(\bSigma_0^2) = (\beta_0^{(0)})^2 +  p^{-1} \tr(\bW_1^2) (\beta_1^{(0)})^2\to_p (\beta_0^{(0)})^2 + \theta(\beta_1^{(0)})^2$.
	For the off-diagonal entries, we shoud have $p^{-1}  \tr(\bSigma_0^2\bW_1) = 2 p^{-1} \tr(\bW_1^2) \beta_0^{(0)}\beta_1^{(0)}+  p^{-1} \tr(\bW_1^3)(\beta_1^{(0)})^2$. By Corollary 2.1.2 of \cite{aguilar2021introduction}, we can show that $p^{-1} \tr(\bW_1^3) \to_p 0$. Then we should have $p^{-1}  \tr(\bSigma_0^2\bW_1) \to_p 2\theta \beta_0^{(0)}\beta_1^{(0)}$. Last, note that $p^{-1}\tr\{(\bSigma_0\bW_1)^2\}= p^{-1} \tr(\bW_1^2) (\beta_0^{(0)})^2 + 2p^{-1} \tr(\bW_1^3) \beta_0^{(0)}\beta_1^{(0)} + p^{-1} \tr(\bW_1^4) (\beta_1^{(0)})^2$. By Corollary 2.1.2 of \cite{aguilar2021introduction}, we can show that $p^{-1} \tr(\bW^4) \to_p 2\theta^2 + \theta$. Then we should have $p^{-1}\tr\{(\bSigma_0\bW_1)^2\}\to_p \theta (\beta_0^{(0)})^2 + ( 2\theta^2 + \theta)(\beta_1^{(0)})^2 $.
	Thus, we obtain that $\bG_{1,p} \to_p \bG_1$ with
	\begin{align*}
		\bG_1 =\begin{bmatrix}
			(\beta_0^{(0)})^2 + \theta(\beta_1^{(1)})^2  & 2\theta\beta_0^{(0)}\beta_1^{(0)}\\
			2\theta\beta_0^{(0)}\beta_1^{(0)} & \theta(\beta_0^{(0)})^2 + ( 2\theta^2 + \theta)(\beta_1^{(1)})^2
		\end{bmatrix}.
	\end{align*}
	It can be verified that the determinant $|\bG_1| >0$, which implies $\bG_1$ is also positive definite. This indicates that Condition (C\ref{cond:convergence}) (i) can hold with high probability.
	
	We next verify Condition (C\ref{cond:convergence}) (ii). Suppose the eigen-decomposition of $\bW_1$
	is $\bW_1 = \bV \bD \bV^\top$, where $\bV$ is an orthogonal matrix, and $\bD$ is a diagonal matrix collecting the eigenvalues of $\bW_1$.
	Then we can derive that,
	\begin{align*}
		&\bSigma_0^{1/2}\bW_1 \bSigma_0^{1/2}= (\beta_0^{(0)}\bI_p +\beta_1^{(0)}\bW_1 )^{1/2} \bW_1 (\beta_0^{(0)}\bI_p +\beta_1^{(0)}\bW_1 )^{1/2} \\
		=&\beta_0^{(0)} \bV \Big\{\bI_p +(\beta_1^{(0)}/\beta_0^{(0)})\bD \Big\}^{1/2} \bV^\top \Big(\bV\bD\bV^\top\Big)  \bV \Big\{\bI_p +(\beta_1^{(0)}/\beta_0^{(0)})\bD \Big\}^{1/2} \bV^\top \\
		=&\beta_0^{(0)} \bV \Big\{\bI_p +(\beta_1^{(0)}/\beta_0^{(0)})\bD \Big\}^{1/2} \bD \Big\{\bI_p +(\beta_1^{(0)}/\beta_0^{(0)})\bD \Big\}^{1/2} \bV^\top \\
		=&\beta_0^{(0)} \bV \Big\{\bD +(\beta_1^{(0)}/\beta_0^{(0)})\bD^2 \Big\} \bV^\top =\beta_0^{(0)} \bW_1 +  \beta_1^{(0)} \bW_1^2.
	\end{align*}
	Consequently, it follows that
	\begin{align*}
		\bH_p =& p^{-1} \begin{bmatrix}
			\tr(\bSigma_0\circ \bSigma_0) & \tr\{(\bSigma_0\circ (\bSigma_0^{1/2}\bW_1 \bSigma_0^{1/2})\}\\
			\tr\{(\bSigma_0\circ (\bSigma_0^{1/2}\bW_1 \bSigma_0^{1/2})\}&  \tr\{(\bSigma_0^{1/2}\bW_1 \bSigma_0^{1/2})\circ (\bSigma_0^{1/2}\bW_1 \bSigma_0^{1/2})\}
		\end{bmatrix}\\
		=&  \begin{bmatrix}
			(\beta_0^{(0)})^2 & p^{-1}\tr(\bW_1^2)  \beta_0^{(0)}  \beta_1^{(0)} \\
			p^{-1}\tr(\bW_1^2)  \beta_0^{(0)}  \beta_1^{(0)} &  p^{-1}\tr(\bW_1^2 \circ \bW_1^2) (\beta_1^{(0)})^2
		\end{bmatrix}.
	\end{align*}
	Recall that $p^{-1}\tr(\bW_1^2)\to_p \theta$. We can also derive that $p^{-1}\tr(\bW_1^2 \circ \bW_1^2) \to_p \theta^2 + \theta$. Then we should have $\bH_p \to_p \bH$ with
	\begin{align*}
		\bH = \begin{bmatrix}
			(\beta_0^{(0)})^2 &  \theta \beta_0^{(0)}  \beta_1^{(0)} \\
			\theta \beta_0^{(0)}  \beta_1^{(0)} &  (\theta^2 + \theta) (\beta_1^{(0)})^2
		\end{bmatrix}.
	\end{align*}
	One can easily verify that the determinant $|\bH| >0$, which implies $\bH$ is also positive definite. This indicates that Condition (C\ref{cond:convergence}) (ii) can also hold with high probability.

	\subsection{Additional Simulation Results}
	\label{append:simu}
	
	In this subsection, we conduct three additional experiments to better evaluate our method. For the first two experiments, we try two different data generation processes of the components of $\bZ$, while holding other simulation settings in Section \ref{subsec:simu} unchanged.
	Specifically, the components of $\bZ$ are assumed to be independently and identically generated from a mixture normal distribution $\xi \cdot \mN(0, 5/9) + (1-\xi)\cdot \mN(0, 5)$ with $P(\xi=1)=0.9$ and $P(\xi=0)=0.1$, or a standardized exponential distribution Exp$(1)-1$.
	The simulation results are presented in Tables \ref{tab:sim_mn}--\ref{tab:sim_se}, respectively.
	For the third experiment, we construct $\bW_k$s with moderate correlation , while generating $\bZ$  from the standard normal distribution and holding other simulation settings in Section \ref{subsec:simu} unchanged.
	Specifically, we independently generate each $\bx_j = (X_{j1}, \dots,X_{jK} )^\top\in \mR^{K}$ $(1\le j\le p)$  from the multivariate normal distribution $\mN_K(\zero, \bSigma_x)$, where $\bSigma_x = (0.5^{|k_1 - k_2|})_{1\le k_1, k_2\le K} \in \mR^{K\times K}$. Then we should have $X_{jk}$s with the same $j$ but different $k$ are linearly correlated with $\textup{corr}(X_{j,k_1}, X_{j, k_2} ) = 0.5^{|k_1 - k_2|}$.  We then construct $\bW_k = (w_{k,j_1j_2})_{1\le j_1, j_2\le p} \in \mR^{p\times p}$ with $w_{k,j_1 ,j_2} = X_{j_1 ,k}  X_{j_2 ,k} \times  \exp\{- p(X_{j_1 ,k}-  X_{j_2 ,k})^2\}$ for each $1\le k\le K$.
	The simulation results are presented in Table \ref{tab:sim_sn_add}.
	By the three tables, we can see that all the results are qualitatively similar to those in Table \ref{tab:sim_sn} of the main text.
	This further demonstrates the robustness and broad applicability of our proposed method.

	\begin{table}[htbp]
		\setlength{\abovecaptionskip}{1.0cm}
		\setlength{\belowcaptionskip}{0.3cm}
		\begin{center}
			\renewcommand\tablename{Table}
			\caption{Simulation results for $\bZ$ generated from the mixture normal distribution.}
			\label{tab:sim_mn}
			\resizebox{\textwidth}{!}{
				\begin{tabular}{cc|ccc|ccc|cc}
					\hline
					$(p,K)$&Penalty&$\TPR$&$\FPR$&$\CS$&$\rmse$&Bias&$\SD$&$\|\cdot\|_2$&$\|\cdot\|_F$\\
					\hline
					\multirow{4}{*}{(200,10)}
					&SCAD    &  0.787 &  0.061 &  0.290 & 0.602 & 0.052 & 0.596 &  8.053 &  2.883 \\
					&MCP     &  0.790 &  0.060 &  0.290 & 0.602 & 0.052 & 0.596 &  8.037 &  2.875 \\
					&OLS     & --     & --     & --     & 0.616 & 0.049 & 0.612 &  8.090 &  3.057 \\
					&ORACLE  &  1.000 &  0.000 &  1.000 & 0.535 & 0.026 & 0.531 &  5.403 &  2.058 \\
					\hline
					\multirow{4}{*}{(500,100)}
					&SCAD    &  0.927 &  0.060 &  0.580 & 0.125 & 0.004 & 0.124 &  6.093 &  1.883 \\
					&MCP     &  0.927 &  0.060 &  0.580 & 0.125 & 0.004 & 0.125 &  6.130 &  1.885 \\
					&OLS     & --     & --     & --     & 0.250 & 0.018 & 0.249 & 19.142 &  5.305 \\
					&ORACLE  &  1.000 &  0.000 &  1.000 & 0.105 & 0.001 & 0.105 &  3.973 &  1.356 \\
					\hline
					\multirow{4}{*}{(1000,1000)}
					&SCAD    &  0.993 &  0.047 &  0.800 & 0.025 & 0.000 & 0.025 &  3.466 &  1.113 \\
					&MCP     &  0.993 &  0.047 &  0.800 & 0.025 & 0.000 & 0.025 &  3.460 &  1.112 \\
					&OLS     & --     & --     & --     & 0.161 & 0.013 & 0.160 & 31.005 & 11.299 \\
					&ORACLE  &  1.000 &  0.000 &  1.000 & 0.022 & 0.000 & 0.022 &  2.482 &  0.878 \\
					\hline
			\end{tabular}}
		\end{center}
	\end{table}

	\begin{table}[htbp]
		\setlength{\abovecaptionskip}{1.0cm}
		\setlength{\belowcaptionskip}{0.3cm}
		\begin{center}
			\renewcommand\tablename{Table}
			\caption{Simulation results for $\bZ$ generated from the standardized exponential distribution.}
			\label{tab:sim_se}
			\resizebox{\textwidth}{!}{
				\begin{tabular}{cc|ccc|ccc|cc}
					\hline
					$(p,K)$&Penalty&$\TPR$&$\FPR$&$\CS$&$\rmse$&Bias&$\SD$&$\|\cdot\|_2$&$\|\cdot\|_F$\\
					\hline
					\multirow{4}{*}{(200,10)}
					&SCAD    &   0.823 &  0.074 &  0.260 & 0.635 & 0.058 & 0.630 &  7.938 &  2.886 \\
					&MCP     &   0.820 &  0.070 &  0.280 & 0.635 & 0.059 & 0.630 &  7.922 &  2.870 \\
					&OLS     &  --     & --     & --     & 0.644 & 0.045 & 0.642 &  8.958 &  3.038 \\
					&ORACLE  &   1.000 &  0.000 &  1.000 & 0.573 & 0.023 & 0.571 &  5.564 &  2.098 \\
					\hline
					\multirow{4}{*}{(500,100)}
					&SCAD    &  0.940 &  0.076 &  0.510 & 0.124 & 0.005 & 0.123 &  5.146 &  1.782 \\
					&MCP     &  0.938 &  0.074 &  0.510 & 0.124 & 0.005 & 0.123 &  5.183 &  1.788 \\
					&OLS     & --     & --     & --     & 0.247 & 0.019 & 0.246 & 15.220 &  5.166 \\
					&ORACLE  &  1.000 &  0.000 &  1.000 & 0.104 & 0.001 & 0.104 &  3.240 &  1.198 \\
					\hline
					\multirow{4}{*}{(1000,1000)}
					&SCAD    &  0.995 &  0.034 &  0.830 & 0.027 & 0.000 & 0.027 &  3.339 &  1.132 \\
					&MCP     &  0.995 &  0.034 &  0.830 & 0.027 & 0.000 & 0.027 &  3.339 &  1.132 \\
					&OLS     & --     & --     & --     & 0.162 & 0.013 & 0.161 & 29.949 & 11.331 \\
					&ORACLE  &  1.000 &  0.000 &  1.000 & 0.025 & 0.000 & 0.025 &  2.757 &  0.973 \\
					\hline
			\end{tabular}}
		\end{center}
	\end{table}

	\begin{table}[htbp]
		{
			\setlength{\abovecaptionskip}{1.0cm}
			\setlength{\belowcaptionskip}{0.3cm}
			\begin{center}
				\renewcommand\tablename{Table}
				\caption{Simulation results for $\bZ$ generated from the standard normal distribution and $\bW_k$s constructed with moderate correlation.}
				\label{tab:sim_sn_add}
				\resizebox{\textwidth}{!}{
					\begin{tabular}{cc|ccc|ccc|cc}
						\hline
						$(p,K)$&Penalty&$\TPR$&$\FPR$&$\CS$&$\rmse$&Bias&$\SD$&$\|\cdot\|_2$&$\|\cdot\|_F$\\
						\hline
						\multirow{4}{*}{(200,10)}
						&SCAD    &  0.588 &  0.103 &  0.060 & 0.793 & 0.164 & 0.748 & 18.883 &  4.172 \\
						&MCP     &  0.575 &  0.115 &  0.050 & 0.830 & 0.182 & 0.776 & 18.925 &  4.222 \\
						&OLS     & --     & --     & --     & 0.833 & 0.062 & 0.826 & 18.902 &  4.398 \\
						&ORACLE  &  1.000 &  0.000 &  1.000 & 0.619 & 0.043 & 0.610 & 15.865 &  3.277 \\
						\hline
						\multirow{4}{*}{(500,100)}
						&SCAD    &   0.745 &  0.054 &  0.160 & 0.210 & 0.021 & 0.155 & 18.136 &  3.615 \\
						&MCP     &   0.733 &  0.051 &  0.150 & 0.218 & 0.023 & 0.150 & 18.234 &  3.679 \\
						&OLS     &  --     & --     & --     & 0.453 & 0.022 & 0.451 & 26.706 &  7.355 \\
						&ORACLE  &   1.000 &  0.000 &  1.000 & 0.118 & 0.004 & 0.115 & 12.488 &  2.322 \\
						\hline
						\multirow{4}{*}{(1000,1000)}
						&SCAD    &   0.845 &  0.093 &  0.280 & 0.066 & 0.002 & 0.039 & 17.189 &  3.281 \\
						&MCP     &   0.848 &  0.087 &  0.320 & 0.068 & 0.003 & 0.038 & 17.068 &  3.311 \\
						&OLS     &  --     & --     & --     & 0.264 & 0.013 & 0.263 & 56.135 & 15.673 \\
						&ORACLE  &   1.000 &  0.000 &  1.000 & 0.024 & 0.000 & 0.024 & 10.051 &  1.751 \\
						\hline
				\end{tabular}}
			\end{center}
		}
	\end{table}

	\begin{table}[htbp]
		\setlength{\abovecaptionskip}{1.0cm}
		\setlength{\belowcaptionskip}{0.3cm}
		\begin{center}
			\renewcommand\tablename{Table}
			\caption{Simulation results for two different tuning parameter selection approaches. Approach (I) is to separately select $\lambda_0$ and $\lambda$, and Approach (II) is to select a single value for both $\lambda_0$ and $\lambda$. }
			\label{tab:selection}
			\resizebox{\textwidth}{!}{
				\begin{tabular}{cc|ccc|ccc|cc}
					\hline
					Approach&Penalty&$\TPR$&$\FPR$&$\CS$&$\rmse$&Bias&$\SD$&$\|\cdot\|_2$&$\|\cdot\|_F$\\
					\hline
					(I) &SCAD    &   0.796 & 0.069 & 0.235 & 0.464 & 0.051 & 0.458 & 7.667 & 2.642 \\
					(II) &SCAD     &   0.792 & 0.070 & 0.230 & 0.465 & 0.053 & 0.459 & 7.732 & 2.656 \\
					(I)&MCP     &  0.796 & 0.070 & 0.230 & 0.464 & 0.051 & 0.458 & 7.690 & 2.645 \\
					(II) &MCP  &   0.794 & 0.071 & 0.220 & 0.465 & 0.053 & 0.459 & 7.730 & 2.656 \\
					\hline
			\end{tabular}}
		\end{center}
	\end{table}

	\subsection{Selection of Tuning Parameters}
	\label{subsec:selection}
	
	To implement the LLA algorithm, we need first compute the Lasso estimator \eqref{eq:lasso} as an initial estimator.
	This requires selecting two tuning parameters: $\lambda_0$ for the Lasso estimator, and $\lambda$ in the folded concave penalized loss function \eqref{eq:Q_lambda}.
	We can separately select the two tuning parameters $\lambda_0$ and $\lambda$.
	However, this approach can be very time-consuming because we need to consider all possible pairs $(\lambda_0, \lambda)$.
	In addition, we can expect that $\lambda \asymp \lambda_0$ as remarked at the end of Appendix \ref{proof of thm:lasso_est}
	Therefore, another approach is to select a single value for both $\lambda_0$ and $\lambda$ by setting $\lambda_0=\lambda$.
	We conducted a preliminary experiment to assess the performance of the two approaches.
	Specifically, we adopt the same simulation setting as in Section \ref{subsec:simu} with $(p,K)=(200,10)$ and $\bZ$ generated from a normal distribution.
	For both approaches, we use the BIC-type criterion \eqref{eq:bic}.
	We replicate the experiment 200 times and compute the same measurements as those in Table \ref{tab:sim_sn}. The results are given in Table \ref{tab:selection}.
	From Table \ref{tab:selection}, we observe that the results of Approach (I) are slightly better than Approach (II). This is expected because Approach (I) explores all possible pairs $(\lambda_0, \lambda)$, while Approach (II) only considers pairs with $\lambda_0 = \lambda$.
	Nevertheless, the two approaches perform very similarly for both the SCAD and MCP estimators.
	In addition, Approach (II) requires less computational time. Consequently, we adopt Approach (II) in the subsequent simulation experiments and real data analysis.
	
}\fi

\end{document}